\documentclass[12pt]{article}
\usepackage[margin=2.5cm]{geometry}
\usepackage{amsmath}
\usepackage{graphicx}
\usepackage[backend=biber,style=apa,citestyle=authoryear]{biblatex}
\usepackage{url} 
\usepackage{threeparttable}
\usepackage{makecell}
\usepackage{multicol}
\usepackage{lineno,hyperref}
\usepackage{float}
\usepackage{amsmath,amssymb,amsthm}
\usepackage{booktabs}
\usepackage{caption}
\usepackage{multirow}
\usepackage{xcolor}
\usepackage[linesnumbered,ruled, vlined]{algorithm2e}
\usepackage{setspace}
\usepackage{appendix}
\usepackage{siunitx}
\usepackage{enumitem}
\usepackage{hyperref}
\usepackage{textcomp}
\usepackage{bbm}
\usepackage{dirtree}
\usepackage{rotating}
\usepackage{psfrag,epsf}
\usepackage{enumerate}
\usepackage{etoolbox}
\usepackage{subcaption}
\theoremstyle{plain}
\newtheorem{theorem}{Theorem}
\newtheorem{proposition}{Proposition}
\newtheorem{lemma}{Lemma}
\newtheorem{corollary}{Corollary}
\theoremstyle{definition}
\newtheorem{assumption}{Assumption}
\newtheorem{definition}{Definition}
\theoremstyle{remark}

\AtBeginEnvironment{algorithm}{\setstretch{1.35}}

\begin{document}
\title{Closely spaced change-point detection with an application to geopolitical shocks in oil and gas markets
}
\author{
Nachiketas Waychal\\
nachiketas.waychal@spjimr.org\thanks{
Corresponding author.
Operations, Supply Chain Management and Quantitative Methods,
S. P. Jain Institute of Management and Research,
Mumbai 400058, India.
}
\and
Arnab Kumar Laha\\
arnab@iima.ac.in\thanks{
Operations and Decision Sciences,
Indian Institute of Management Ahmedabad,
Ahmedabad 380015, India.
}
\and
Ankur Sinha\\
asinha@iima.ac.in\thanks{
Operations and Decision Sciences,
Indian Institute of Management Ahmedabad,
Ahmedabad 380015, India.
}
}
\maketitle
\begin{abstract}
Change points that arrive in quick succession, or near the boundaries of the sample, are prone to type-II error: the few observations trapped between successive breaks carry too little information for global procedures to analyze. This article develops a test statistic, $I_1$, for which we derive the limiting null distribution under Gaussian noise, extend the theory to weakly dependent sequences, and prove consistency under a mean-shift alternative. Around $I_1$, we construct DCI, a two-phase algorithm that couples a global scan with multiple local tests via an adaptive-width sliding window. In extensive simulations, $I_1$ and DCI remain Pareto-optimal over criteria spaces spanning true and false positive rates, localization error, and segmentation covering. An application to the US-Iran 2026 conflict quantifies the exposure of energy markets to geopolitical shocks using two new metrics, event sensitivity and event attribution, leading to time series tokenization.
\end{abstract}
\newpage
\section{Introduction}
\label{sec:intro}

Change point detection, the identification of time points at which the statistical properties of a series shift abruptly, is a central problem in time series analysis. In finance, economics, engineering, and environmental studies alike, a reliable account of structural shifts is a prerequisite for explainable modeling, accurate forecasting, and sound decision-making, whether the object of interest is a market regime change or the impact of a macroeconomic policy. Modern data environments, however, increasingly confront the analyst with series in which change points occur in quick succession and separated by too few observations. Volatile, high-frequency financial data during geopolitical crises are a leading example. Conventional detection methods, designed with well-separated breaks in mind, often struggle in precisely these settings, and this difficulty motivates the present article.

The literature on change point analysis is vast, and we can only outline the strands most relevant to our work. A classical formulation casts change detection as a model selection problem \parencite{truong2020selective}, resolved through likelihood ratio maximization, with foundational contributions by \textcite{hinkley1970inference} and multivariate mean change tests developed by \textcite{srivastava1986likelihood}. In parallel, nonparametric statistics, most prominently the CUSUM of \textcite{page1954continuous} and its extensions \parencite{fremdt2015page,han2023optimization,inclan1994use}, dispense with distributional assumptions and offer model-free alternatives. Within this tradition, solutions are now available for changes in the mean \parencite{pettitt1979non}, high-dimensional means \parencite{drikvandi2025distribution}, the slope \parencite{fearnhead2019detecting}, the variance \parencite{qin2017strong}, and the covariance structure \parencite{londschien2021change}. Further nonparametric approaches include rank- and U-statistic-based procedures for high-dimensional data \parencite{levy2009detection}, graph-based methods \parencite{chen2015graph}, and multiple change-point detection in multivariate data \parencite{matteson2014nonparametric}. For the multiple change point problem specifically, exact search procedures based on dynamic programming trace back to the segment neighborhood algorithm of \textcite{auger1989algorithms}, with computationally efficient descendants in the pruned exact linear time (PELT) method \parencite{killick2012optimal}, CPOP \parencite{fearnhead2019detecting}, and RFPOP \parencite{fearnhead2019changepoint}. Nonparametric maximum-likelihood formulations of the multiple change problem typically select the number of breaks by minimizing a penalty function \parencite{zou2014nonparametric}. A second prominent family builds on binary segmentation, a `greedy' strategy that ignores the global configuration at each split; refinements include circular binary segmentation \parencite{olshen2004circular}, wild binary segmentation \parencite{fryzlewicz2014wild}, sparsified binary segmentation \parencite{cho2015multiple}, and seeded binary segmentation \parencite{kovacs2023seeded}, whereas `generous' bottom-up alternatives are reviewed by \textcite{truong2020selective}. Additional proposals of practical importance include TrendSegment \parencite{maeng2024detecting}, the change point procedure via pruned objectives \parencite{zhang2017pruning}, and the narrowest-over-threshold method \parencite{baranowski2019narrowest}, which blends nonparametric and parametric elements sequentially. The MOSUM procedure of \textcite{eichinger2018} detects multiple changes in the mean by scanning standardized differences between sample means computed over adjacent moving windows and identifying threshold-exceeding local maxima as change point estimates. Although the proposed $I_1$ and DCI also deploy a sliding window, they differ from MOSUM in both design and performance. Comparison of the proposed and existing methods (Section~\ref{sec:sims}) shows that the proposals outperform the MOSUM procedure across a range of scenarios. Finally, kernel-based procedures \parencite{arlot2019kernel} and supervised learning methods, largely classification-based and driven by sensor data for activity recognition, are surveyed by \textcite{aminikhanghahi2017survey}. Common to nearly all of these proposals, however, is the implicit assumption that successive change points are separated by segments long enough to permit reliable estimation on either side.

Three considerations distinguish our contribution from this literature. First, methods for multiple change points typically privilege a global view of the series; this is necessary for comprehensive detection, but might compromise localized insight when breaks arrive in quick succession due to the near-insufficient number of observations between successive change points, inflating type-II error. Second, the computational burden of existing procedures grows rapidly with the number of candidate breaks and the forms of change, limiting scalability. Third, breaks near the extremities of the sample remain notoriously difficult to detect. We address all three issues by proposing a new test statistic, $I_1$, whose asymptotic properties and limiting null distribution we establish, and a two-phase algorithm, the Detector of Change-points using $I_1$ (DCI), which reconciles a global scan with multiple local tests. A new design feature, an adaptive-width sliding window, enables $I_1$ and DCI to retain high true detection rates under the afore-stated irregular configurations.

The remainder of the article is organized as follows. Section~\ref{sec:statasympt} states the model and assumptions, introduces $I_1$, derives its asymptotic properties and limiting null distribution, and presents the DCI algorithm for closely-spaced multiple change points detection. Section~\ref{sec:sims} reports simulation studies for $I_1$ and DCI separately, spanning independent and identically (Normally) distributed (i.i.d.) and autoregressive processes, zero and non-zero slopes, and distant as well as rapidly occurring change points; across these settings the proposals are Pareto-optimal with respect to true positive rate, false positive rate, mean absolute error (MAE), and the (unweighted) segmentation covering metric. Section~\ref{sec:app} provides an empirical application assessing the impact of the US-Iran conflict on oil and gas markets, illustrating tokenization (dimensionality reduction) of financial time series for further predictive analytics. Section~\ref{sec:conc} concludes with a discussion on future studies that can extend $I_1$ for detecting structural breaks in higher moments of temporal processes.

\section{Test Statistic and Asymptotic Properties}
\label{sec:statasympt}
Consider a temporally ordered sequence of observations $y_{1_t}$, $t = 1, 2, \ldots, T$, which may be either i.i.d. or weakly dependent, and let $y_t = (y_{1_t} - \bar{y_1})/\sigma_{y_1}$ denote its standardized version; where $\bar{y_1}$ and $\sigma^2_{y_1}$ are the mean and variance of $y_{1_t}$ (refer Appendix~\ref{ssec:varest} for robust variance estimation in the presence of change points). Our objective is to construct a statistic that is sensitive to local structural changes in the mean of $y_t$ (and, consequently, of any $y_{1_t}$ with unknown mean and variance) while remaining analytically tractable under the no-change null hypothesis and consistent asymptotically. The construction proceeds in two stages. We first form a local future-minus-past contrast $d_t$; these local contrasts are then aggregated symmetrically to yield a transformed sequence $z_t$. The test statistic $I_1$ is defined as the maximum absolute value of this transformed sequence.
\subsection{Model and Assumptions}
\label{ssec:modassm}
The theoretical development begins under the i.i.d. Gaussian assumption, for which we derive the distribution of $I_1$ under $H_0$. Having established this baseline null theory, we extend the construction to short-range dependent stationary sequences. Finally, under a mean-shift alternative, we show that a structural change in $y_t$ induces a localized deterministic peak in $z_t$, and that the global test based on $I_1$ is consistent whenever this induced peak dominates the null extreme-value threshold.
\begin{assumption}\label{as:iid}
Under the null hypothesis of no change, $H_0:\ y_t \overset{\text{i.i.d.}}{\sim} N(0,1)$, $t=1,\dots,T$.
\end{assumption}
\begin{assumption}\label{as:dep}
Under $H_0$, suppose that $\{y_t\}$ is strictly stationary with $E(y_t)=0$, $\mathrm{Var}(y_t)=1$. Assume that $\{y_t\}$ is weakly dependent with short-range dependence. Specifically, suppose that either of the following holds: $\{y_t\}$ is a stationary Gaussian sequence with autocovariance function $\gamma_y(k)=\mathrm{Cov}(y_t,y_{t+k})$ satisfying $|\gamma_y(k)|\le C\rho^{|k|}$, $0<\rho<1$; $\{y_t\}$ is strong mixing \parencite{rosenblatt1956} with coefficients $\alpha_y(k)$ satisfying $\alpha_y(k)=O(\rho^{k})$, $0<\rho<1$.
\end{assumption}
\begin{assumption}\label{as:alt}
Under the alternative hypothesis $H_1$, $y_t=\mu_t+\varepsilon_t$, where $\{\varepsilon_t\}$ satisfies either Assumption~\ref{as:iid} or Assumption~\ref{as:dep}, and $\mu_t$ contains at least one mean shift. For the single-change case, let $\mu_t=\mu_1\mathbf 1\{t\le\eta\}+\mu_2\mathbf 1\{t>\eta\}$ where
\begin{equation}
\delta=|\mu_2-\mu_1|\neq 0. \label{eq:delta}
\end{equation}
\end{assumption}
\begin{definition}\label{def:stat}
Let $w$ denote the sliding-window parameter \parencite{eichinger2018} such that
\begin{equation}
\lim_{N\to\infty}\frac{w}{\log N}=K \label{eq:K}
\end{equation}
where $K\in(0,\infty)$ is a finite positive constant. For $t=w,\dots,T-w$, define the local future-minus-past contrast
\begin{equation}
d_t=\left(\frac{1}{w}\sum_{i=1}^{w}y_{t+i}\right)-\left(\frac{1}{w}\sum_{i=1}^{w}y_{t-i+1}\right). \label{eq:dt}
\end{equation}
Define an admissible index set $\mathcal{I}_T=\{2w-1,\dots,T-2w+1\}$ with cardinality $N=|\mathcal{I}_T|=T-4w+3$. For $t\in\mathcal{I}_T$, define the symmetrically aggregated transformed sequence
\begin{equation}
z_t=\sum_{j=t-w+1}^{t+w-1}d_j=\sum_{r\in\mathcal{C}_w}c_{r,w}\,y_{t+r} \label{eq:zt}
\end{equation}
where $\mathcal{C}_w=\{-2w+2,\dots,2w-1\}$ and
\begin{equation}
c_{r,w}=
\begin{cases}
-\dfrac{r+2w-1}{w}, & -2w+2\le r\le -w+1,\\[4pt]
\dfrac{r-1}{w}, & -w+2\le r\le 0,\\[4pt]
\dfrac{r}{w}, & 1\le r\le w,\\[4pt]
\dfrac{2w-r}{w}, & w+1\le r\le 2w-1,\\[4pt]
0, & \text{otherwise.}
\end{cases}
\label{eq:crw}
\end{equation}
The proposed test statistic is
\begin{equation}
I_1=\max_{t\in\mathcal{I}_T}|z_t|. \label{eq:I1}
\end{equation}
\end{definition}
\subsection{Asymptotic properties of $I_1$}
This section presents the propositions, lemmas, theorems, and corollaries establishing the asymptotic properties of $I_1$ under Extreme Value Theory (EVT). Proofs of all formal results are provided in Appendix~\ref{apsec:proofs}.
\label{ssec:asympti1}
\begin{proposition}\label{prop:stat}
Under Assumption~\ref{as:iid} and Definition~\ref{def:stat}, the sequence $z_t$ is a zero-mean stationary Gaussian sequence. In particular, $E(z_t)=0$ and $\sigma_z^2=\mathrm{Var}(z_t)=\frac{4w}{3}+\frac{2}{3w}$. Consequently, $\{|z_t|\}$ is an $m$-dependent sequence \parencite{watson1954} with $m=4w-3$.
\end{proposition}
\begin{lemma}\label{lem:berman}
Under the regular conditions established in Assumption~\ref{as:iid} and Proposition~\ref{prop:stat}, $\{z_t\}$ directly satisfies Berman's condition \parencite{berman1964}, $\lim_{k\to\infty}\gamma_z(k)\log k=0$, thus satisfying Leadbetter's $D(u_N)$ and $D'(u_N)$ conditions \parencite{leadbetter1974,llr1983} for the threshold sequence $u_N=a_Nx+b_N$, and yielding an asymptotic extremal index of unity $\theta=1$ \parencite{leadbetter1983extremal}. Consequently, the absolute sequence $\{|z_t|\}$ also satisfies Berman's condition and exhibits $\theta=1$.
\end{lemma}
\begin{theorem}\label{thm:gumbel}
As the sample size diverges ($T\to\infty$, implying $N\to\infty$), the normalized test statistic $I_1$ converges weakly to a two-sided Gumbel extreme value distribution \parencite{fisher1928,gnedenko1943}:
\begin{equation}
\lim_{N\to\infty}P_{H_0}\!\left(\frac{I_1-b_N}{a_N}\le x\right)=\exp\!\left(-2e^{-x}\right),\qquad \forall x\in\mathbb{R}. \label{eq:limit}
\end{equation}
Furthermore, the asymptotic mean and variance of the statistic $I_1$ are given by:
\begin{align}
E[I_1]&=b_N+a_N(\gamma+\log 2)+o(a_N) \label{eq:mean}\\
\mathrm{Var}(I_1)&=a_N^2\,\frac{\pi^2}{6}+o(a_N^2) \label{eq:var}
\end{align}
where $\gamma\approx 0.5772$ denotes the Euler-Mascheroni constant \parencite{gumbel1958}.
\end{theorem}
\begin{corollary}\label{cor:crit}
The asymptotic level-$\alpha$ critical value for $I_1$ such that $\lim_{N\to\infty}P_{H_0}(I_1>\lambda_\alpha)=\alpha$, is
\begin{equation}
\lambda_\alpha=b_N+a_N\left(-\log\left(-\tfrac{1}{2}\log(1-\alpha)\right)\right). \label{eq:crit}
\end{equation}
\end{corollary}
\begin{lemma}\label{lem:mixing}
Under the relaxed dependence structure in Assumption~\ref{as:dep} and the transformation rules in Definition~\ref{def:stat}, the aggregated sequence $z_t=\sum_{r\in\mathcal{C}_w}c_{r,w}y_{t+r}$ inherits strict stationarity from the underlying process $\{y_t\}$. Furthermore, the strong-mixing property is preserved under this finite-filter transformation \parencite{bradley2005}. For any temporal separation $k$ exceeding the compound filter width, the mixing coefficients $\alpha_z(k)$ of $z_t$ satisfy:
\begin{equation}
\alpha_z(k)\le\alpha_y(k-4w+3),\qquad\text{for all }k>4w-3. \label{eq:mix}
\end{equation}
Consequently, $\alpha_z(k)\to 0$ as $k\to\infty$ at an exponential rate. Under serial correlation, the variance of the transformed sequence is given by:
\begin{equation}
\sigma_{z,\mathrm{dep}}^2=\mathrm{Var}(z_t)=\sum_{r\in\mathcal{C}_w}\sum_{s\in\mathcal{C}_w}c_{r,w}c_{s,w}\gamma_y(s-r), \label{eq:vardep}
\end{equation}
where $\gamma_y(k)=\mathrm{Cov}(y_t,y_{t+k})$. Because the mixing coefficients decay rapidly as stated in Assumption~\ref{as:dep}, the double sum is stable and positive, scaling linearly with the window width: $\sigma_{z,\mathrm{dep}}^2=O(w)$. The long-run variance ensures that the extremal index $\theta$ of the resulting sequence remains equal to unity.
\end{lemma}
\begin{theorem}\label{thm:dep}
Suppose Assumption~\ref{as:dep} holds. Let $I_1$ be defined as in Equation \eqref{eq:I1}, and let $a_N$ and $b_N$ be defined as in Equation \eqref{eq:norm}, with $\sigma_z^2=\mathrm{Var}(z_t)$. Then, for every $x\in\mathbb{R}$,
\begin{equation}
\lim_{N\to\infty}P_{H_0}\!\left(\frac{I_1-b_N}{a_N}\le x\right)=\exp\{-2e^{-x}\}. \label{eq:limitdep}
\end{equation}
Consequently, the critical value $\lambda_\alpha$ in Equation \eqref{eq:crit} continues to satisfy
\begin{equation}
\lim_{N\to\infty}P_{H_0}(I_1>\lambda_\alpha)=\alpha. \label{eq:leveldep}
\end{equation}
\end{theorem}
\begin{lemma}\label{lem:signal}
Under Assumption~\ref{as:alt} and Definition~\ref{def:stat}, decompose $z_t=m_z(t)+\xi_t$ for $t\in\mathcal{I}_T$, where $m_z(t)$ and $\xi_t$ represent the deterministic and stochastic components, respectively, of $z_t$. Then there exists a deterministic filter-response function $H_w(\cdot)$, depending only on the transformation coefficients, such that $m_z(t)=\delta H_w(t-\eta)$. Furthermore, $H_w(\cdot)$ is localized around the change point and $\max_u|H_w(u)|=w$. Consequently, $\max_t|m_z(t)|=w\delta$. Thus, a mean shift of size $\delta$ in $y_t$ induces a localized peak in $|z_t|$ whose deterministic height is of the order $O(w\delta)$.
\end{lemma}
\begin{theorem}\label{thm:power}
Let $\lambda_\alpha$ be the asymptotic critical value in Equation \eqref{eq:crit}. Under Assumption~\ref{as:alt} and Lemma~\ref{lem:signal}, let $M_N=w\delta$ denote the maximum deterministic response induced in $|z_t|$. Then for every finite $\delta>0$, there exists a uni-directionally unbounded range $K\in\left(\frac{8}{3\delta^2},\infty\right)$ such that:
\begin{equation}
\lim_{N\to\infty}P_{H_1}(I_1>\lambda_\alpha)\to 1. \label{eq:power}
\end{equation}
\end{theorem}
To enable $I_1$ to detect change points when fewer than $2w$ observations are available in the neighborhood of a change point, we deploy an adaptive-width sliding-window enhancement (Appendix~\ref{ssec:adapwid}). Further, while the asymptotic distribution of $I_1$ remains the same (Equation~\ref{eq:limit}) for both i.i.d. and weakly dependent processes, finite-sample distributions for real-world applications with limited data may vary. We simulate the empirical distributions of $I_1$ under $H_0$ for a range of autocorrelation ($\rho$) values. A user can select one of the two modes of $I_1$ - `default' (uses simulated distribution of $I_1$ for $\rho=0$) or `WD' (identifies $\rho$ of the `weakly dependent' data and uses simulated distribution for the closest rounded-off value).

\subsection{An algorithm for detecting closely-spaced multiple change points}
\label{ssec:dci}
This section introduces the proposed \textbf{D}etector of \textbf{C}hange-points using \(I_1\) (DCI), an algorithm for the rapid detection of multiple change points of different forms in a univariate time series. The algorithm proceeds in two phases: Phase I, \textit{Flagging}, and Phase II, \textit{Evaluation}. Phase I performs a global scan of \(y_t\) and flags potential change point locations throughout the series. To reduce the false-negative rate at this stage, we introduce an enhancement called \(p\)-local scaling; its details are provided in Appendix~\ref{plocal}. The complete specification of Phase I is given in Algorithm~\ref{algo:seg}. Let \(k\) denote the number of potential change points flagged in Phase I. Phase II, \textit{Evaluation}, then proceeds through \(k\) iterations, with each iteration evaluating exactly one flagged location. The details of Phase II are given in Algorithm~\ref{algo:eval}.
\begin{algorithm}
\footnotesize
\SetKwInOut{KwInput}{Input}
\SetKwInOut{KwOutput}{Output} 
\SetKwFunction{Centroid}{centroid}
\SetKwFunction{Sign}{sgn}
\KwInput{Time series $y_t$ ($t \in \{1, 2, \dots, T\}$); Hyperparameters: $p, w, r_1$}
Calculate $\{d_t\}_{t=1}^{T-1}$ using adaptive-width sliding window (Equation \eqref{eqadapwin})\;
$\{dp_{t}\}_{t=1}^{T-1} \leftarrow p$-locally scaled IQR exceedances of $\{d_t\}_{t=1}^{T-1}$ (Equations~\eqref{eqplocal1}~\eqref{eqplocal2})\;
Partition the index set $\{1, \dots, T-1\}$ into maximal contiguous subsets $\mathcal{D}_1, \dots, \mathcal{D}_K$ such that $\forall k \in \{1, \dots, K\}$ and $\forall t, t' \in \mathcal{D}_k$, $\Sign(dp_t) = \Sign(dp_{t'})$\;
Initialize $flag\_weight_t \leftarrow 0$ for all $t \in \{1, 2, \dots, T\}$\;
\For{$k \leftarrow 1$ \KwTo $K$}{
    $c_k \leftarrow \Centroid(\mathcal{D}_k)$\;
    $flag\_weight_{c_k} \leftarrow \left|\sum_{t \in \mathcal{D}_k} dp_t\right|$\;
}
$\boldsymbol{\hat{\eta}} \leftarrow \{t \in \{1, \dots, T\} \mid flag\_weight_t > r_1\}$\;
\KwOutput{$\boldsymbol{\hat{\eta}}$}
\caption{DCI Phase-1: Flagging}
\label{algo:seg}
\end{algorithm}
\begin{algorithm}
\footnotesize
\SetKwInput{KwInput}{input}                
\SetKwInput{KwOutput}{output}              
\SetKwFunction{Union}{Union}
\SetKwFunction{Kwdiff}{Difference}
\SetKwFunction{Labl}{Label}
\KwInput{$\boldsymbol{\hat{\eta}}$, $\alpha$, $\lambda_\alpha$}
Initiate $I^{[j]}$ for $j \in \{ 1,\cdots k \}$\\
\For{$j \leftarrow 1 $ \KwTo $k$ }{
    $y^{[j]}_t = \langle y_{\hat{\eta}_{j-1}+1},\cdots,y_{\hat{\eta}_{j+1}} \rangle$\\
    
        $I^{[j]} \leftarrow I_1( y^{[j]}_t )$ \\
        \eIf{$I^{[j]} < \lambda_\alpha$}{
        $\boldsymbol{\hat{\eta}} \leftarrow$ \Kwdiff{$\boldsymbol{\hat{\eta}}, \hat{\eta}_j$}}{
        $\hat{\eta}_j \leftarrow \hat{\eta}^*(y^{[j]}_t)$
        }
         
}
\KwOutput{$\boldsymbol{\hat{\eta}^*}$}
\caption{DCI Phase-2: Evaluation}
\label{algo:eval}
\end{algorithm}
\section{Simulations}
\label{sec:sims}
In this section, we conduct simulation studies to investigate the finite sample performance of the proposed methods under the null and alternative hypotheses. We first evaluate the test statistic $I_1$ in the single change point setting against a broad collection of established methods (Section~\ref{ssec:pfmi1}), and then assess the DCI algorithm in the multiple change point regime for which it is designed (Section~\ref{ssec:pfmdci}). Throughout, performance is judged over a multi-objective criteria space rather than by any single summary measure.
\subsection{Performance Evaluation of $I_1$}
\label{ssec:pfmi1}
To evaluate the performance of $I_1$ in detecting a level change and to compare it with existing change-point detection tests, we conduct a comprehensive simulation exercise. Each simulated series contains at most one level-change point, with both the location and the magnitude of the change varying across replications. The noise behavior spans i.i.d. and weakly dependent settings, with autocorrelation ($\rho$) values uniformly distributed across the positive and negative ranges. The magnitude of the level change is uniformly distributed over $\frac{\delta}{\sigma} \in \{0, 0.5, 1, 2, 3, 4, 5, 6\}$, and the number of observations available after the change point is uniformly distributed over $T - \eta \in \{1, 2, 3, \ldots, 22, 50\}$. The experiment is conducted with $\alpha=0.1$ and $T=100$. By crossing weak and pronounced shifts with post-change samples ranging from a single observation to a comfortable margin, under both independent and serially dependent noise, the design spans the spectrum from favorable to near-degenerate detection conditions and thus provides a robust basis for assessing and comparing change point detection tests.

The evaluation criteria are as follows. The true positive rate (TPR) quantifies aggregate detection performance across the $\delta$ and $\eta$ ranges; when computed for an individual $(\delta, \eta)$ pair, it coincides with the test's power. The false positive rate (FPR) is the percentage of $H_0$ instances, that is, series without an embedded change point, in which the test rejects $H_0$. The mean absolute error (MAE) is the mean absolute temporal distance between $\eta$ and $\hat{\eta}$, where in type-II error situations $\hat{\eta}$ is set to $t = 0$.

Over this multi-objective space, $I_1$ (`WD' mode) is compared with popular methods for level (mean) change point detection: the cumulative sum (CUSUM) statistic \parencite{page1954continuous}, the change point procedure via pruned objectives (CP3O) \parencite{zhang2017pruning}, kernel change point analysis (KCPA) \parencite{arlot2019kernel}, Bayesian change point analysis (BCP) \parencite{wang2015bayesian}, the nonparametric test of \textcite{pettitt1979non}, and the likelihood ratio test (LRT) statistic \parencite{hinkley1970inference}. The likelihood ratio statistic is implemented via the function \texttt{cpt.mean} from the R package \texttt{changepoint} \parencite{killick2014changepoint}, and the CUSUM statistic via the R package \texttt{CPAT} \parencite{miller2018cpat}. The KCPA and CP3O methods are implemented from the R package \texttt{ecp} \parencite{james2015ecp}, the Pettitt test from the R package \texttt{trend} \parencite{pohlert2020trend}, and the BCP method from the R package \texttt{bcp} \parencite{wang2018package}.

Table~\ref{tab:pfmi1} reports the performance comparison over the stated criteria space, in which $I_1$ is evaluated against seven benchmark tests. $I_1$ dominates all the competitors in the TPR-MAE space, and lies on the Pareto-optimal frontier in the TPR-FPR space. Figure~\ref{fig:ppi1tprfpr} displays this frontier, with the shaded region representing the space dominated by $I_1$. Furthermore, $I_1$ outperforms the MOSUM procedure in the TPR-FPR-MAE space at the aggregate level.

Figures~\ref{fig:i1tprmagarneg}, \ref{fig:i1tprmagarpoz}, \ref{fig:i1tpreaarneg}, and \ref{fig:i1tpreaarpoz} show the variation in TPR as the magnitude of change $\frac{\delta}{\sigma}$ and the location of the embedded change point $\eta$ vary, where $T - \eta$ measures the number of observations available post the change point. Small values of $T - \eta$ correspond to near-terminal change points, and it is precisely in these challenging scenarios that $I_1$ outperforms most of the benchmark tests. Figures~\ref{fig:i1maemagarneg}, \ref{fig:i1maemagarpoz}, \ref{fig:i1maeetaarneg}, and \ref{fig:i1maeetaarpoz} present the analogous variation in MAE, across which $I_1$ remains in the top-performing (low MAE) zones.

\begin{table}
\centering
\small
\setlength{\tabcolsep}{5pt}
\begin{tabular}{lccccccccc}
\toprule
Metric & $\rho$ 
& $I_1$ & BCP & CP3O & CUSUM & KCPA & LRT & MOSUM & Pettitt \\
\midrule

\multirow{4}{*}{TPR}
 & $<0$ & 79\% & 47\% & 17\% & 75\% & 65\% & 75\% & 70\% & 64\% \\
 & $=0$ & 78\% & 47\% & 18\% & 77\% & 67\% & 75\% & 73\% & 66\% \\
 & $>0$ & 76\% & 48\% & 18\% & 79\% & 68\% & 76\% & 77\% & 68\% \\
 & Aggregate & 78\% & 47\% & 18\% & 77\% & 67\% & 75\% & 73\% & 66\% \\
\midrule

\multirow{4}{*}{FPR}
 & $<0$ & 8\%  & 0\% & 0\% & 10\% & 20\% & 0\% & 5\%  & 7\% \\
 & $=0$ & 8\%  & 0\% & 0\% & 12\% & 19\% & 0\% & 13\% & 10\% \\
 & $>0$ & 10\% & 2\% & 0\% & 17\% & 24\% & 0\% & 25\% & 14\% \\
 & Aggregate & 9\%  & 1\% & 0\% & 13\% & 21\% & 0\% & 14\% & 10\% \\
\midrule

\multirow{4}{*}{MAE}
 & $<0$ & 21.1 & 46.3 & 75.4 & 25.5 & 33.9 & 23.1 & 27.6 & 38.8 \\
 & $=0$ & 22.6 & 46.1 & 75.1 & 24.6 & 32.9 & 22.9 & 26.2 & 37.9 \\
 & $>0$ & 24.2 & 45.6 & 75.0 & 23.9 & 32.5 & 22.6 & 24.5 & 37.1 \\
 & Aggregate & 22.6 & 46.0 & 75.2 & 24.7 & 33.1 & 22.9 & 26.1 & 37.9 \\
\bottomrule
\end{tabular}
\caption{Performance of $I_1$ across different levels of autocorrelation ($\rho$) in the simulated data}
\label{tab:pfmi1}
\end{table}

\begin{figure}
    \centering
    \begin{subfigure}[b]{0.32\textwidth}
        \centering
        \includegraphics[width=\textwidth]{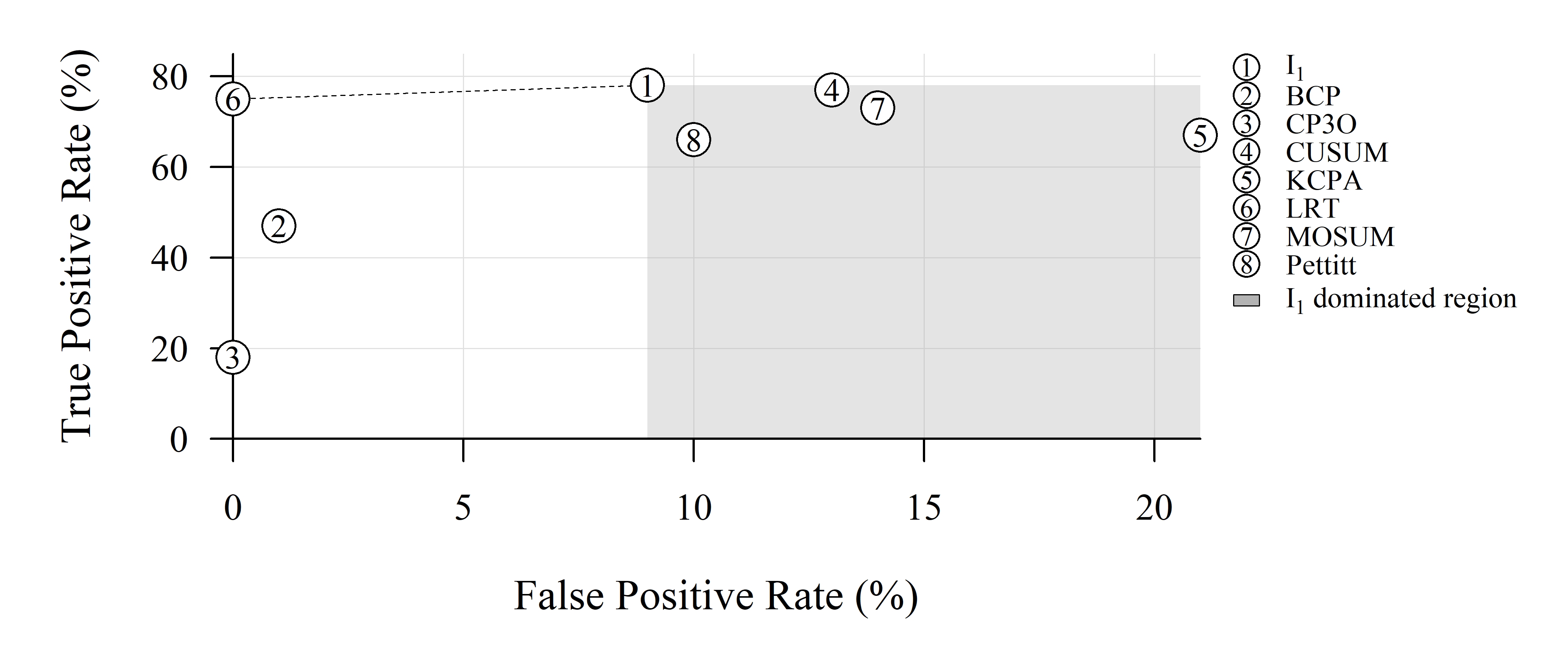}
        \caption{TPR vs FPR: Aggregate}
        \label{fig:ppi1tprfpr}
    \end{subfigure}
    \begin{subfigure}[b]{0.32\textwidth}
        \centering
        \includegraphics[width=\textwidth]{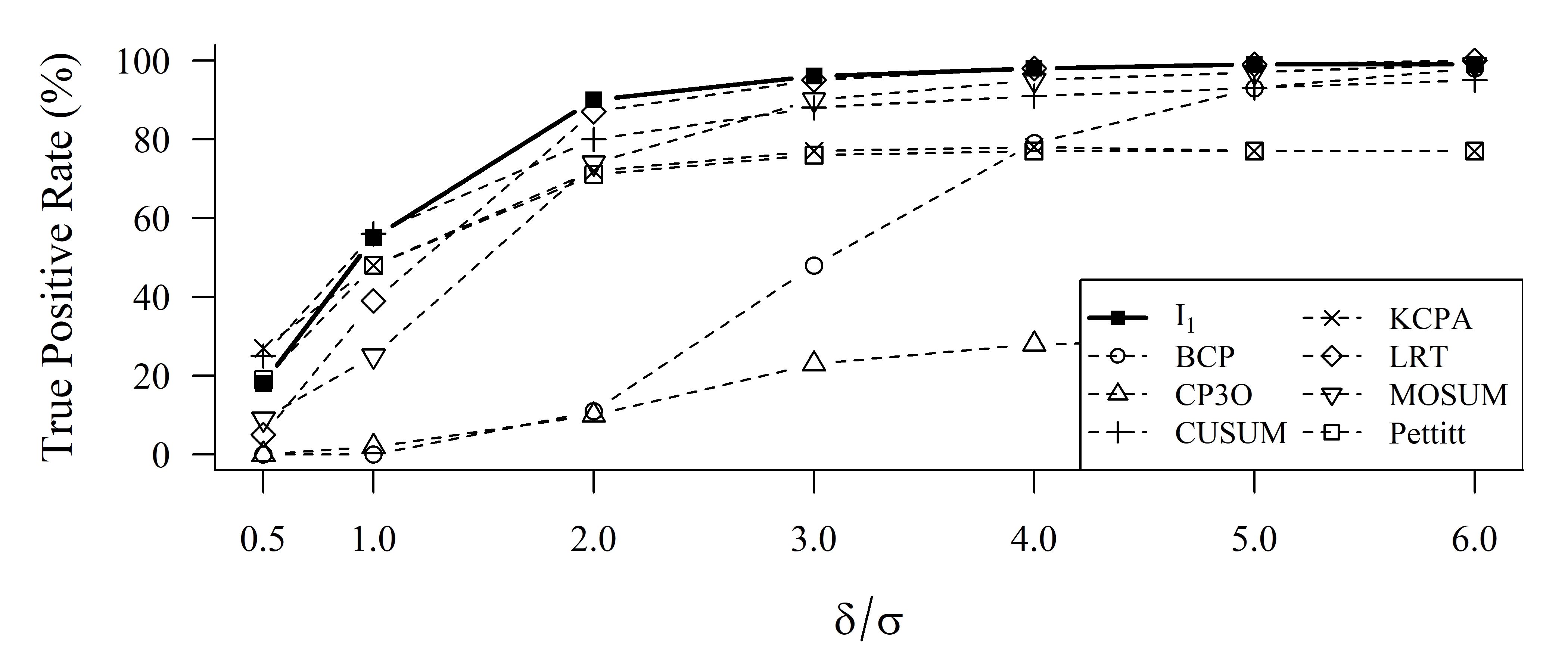} 
        \caption{TPR vs $\frac{\delta}{\sigma}$ ($\rho<0$)}
        \label{fig:i1tprmagarneg}
    \end{subfigure}
    \begin{subfigure}[b]{0.32\textwidth}
        \centering
        \includegraphics[width=\textwidth]{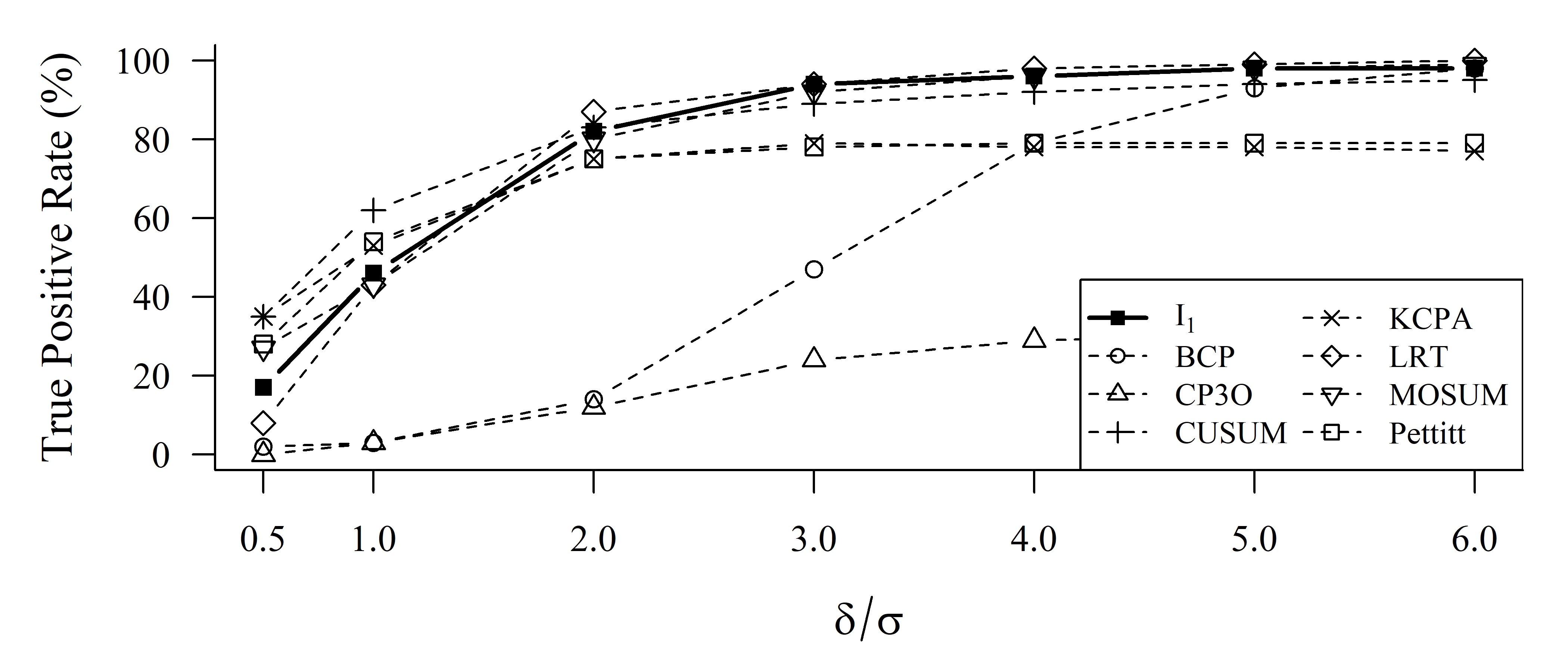}
        \caption{TPR vs $\frac{\delta}{\sigma}$ ($\rho>0$)}
        \label{fig:i1tprmagarpoz}
    \end{subfigure}

    \vspace{1em} 

    \begin{subfigure}[b]{0.32\textwidth}
        \centering
        \includegraphics[width=\textwidth]{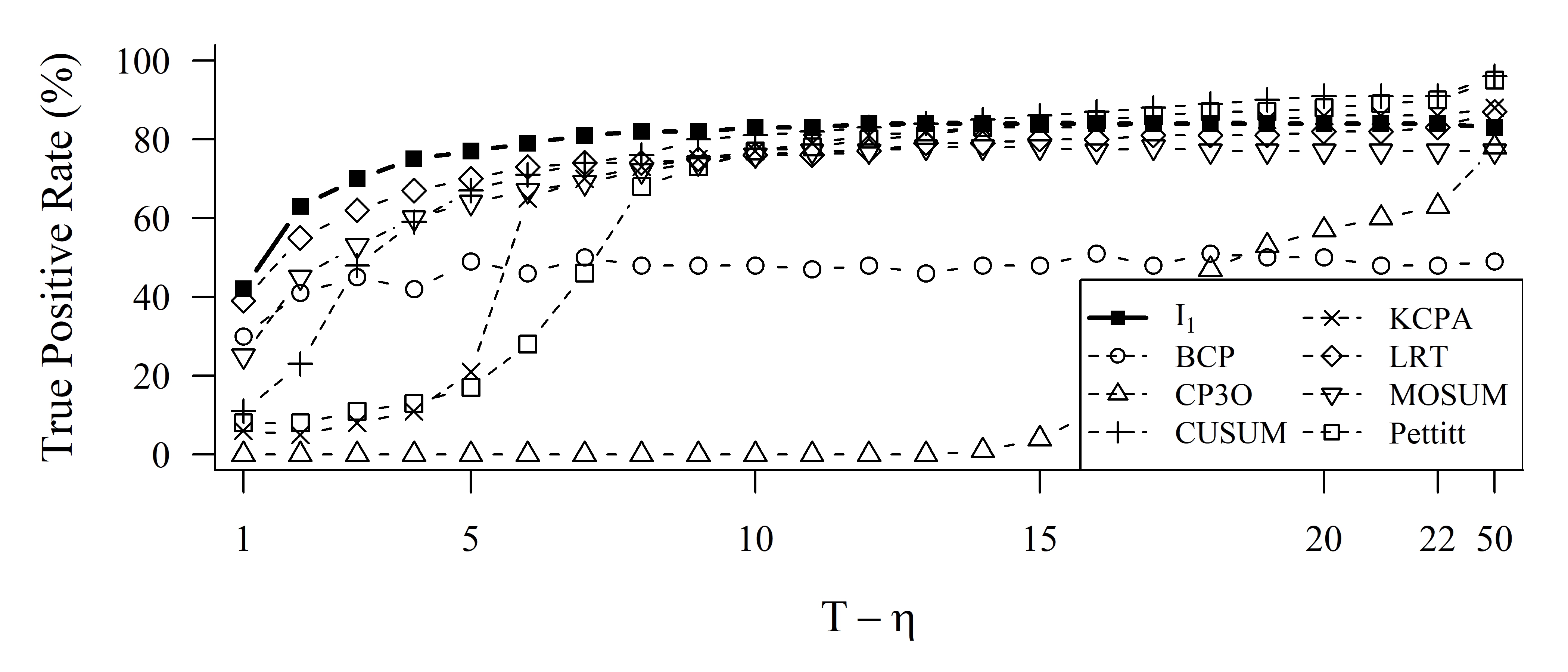}
        \caption{TPR vs $T-\eta$ ($\rho<0$)}
        \label{fig:i1tpreaarneg}
    \end{subfigure}
    \begin{subfigure}[b]{0.32\textwidth}
        \centering
        \includegraphics[width=\textwidth]{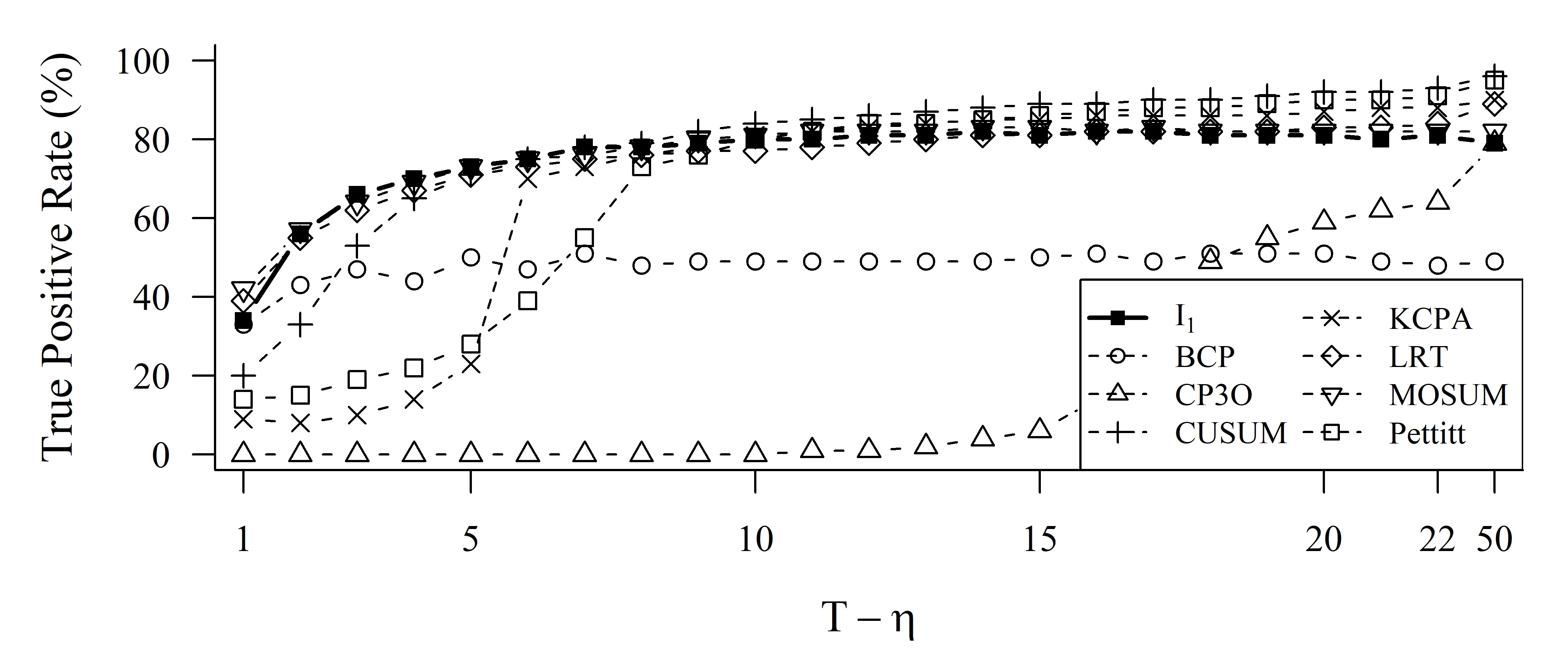}
        \caption{TPR vs $T-\eta$ ($\rho>0$)}
        \label{fig:i1tpreaarpoz}
    \end{subfigure}
    \begin{subfigure}[b]{0.32\textwidth}
        \centering
        \includegraphics[width=\textwidth]{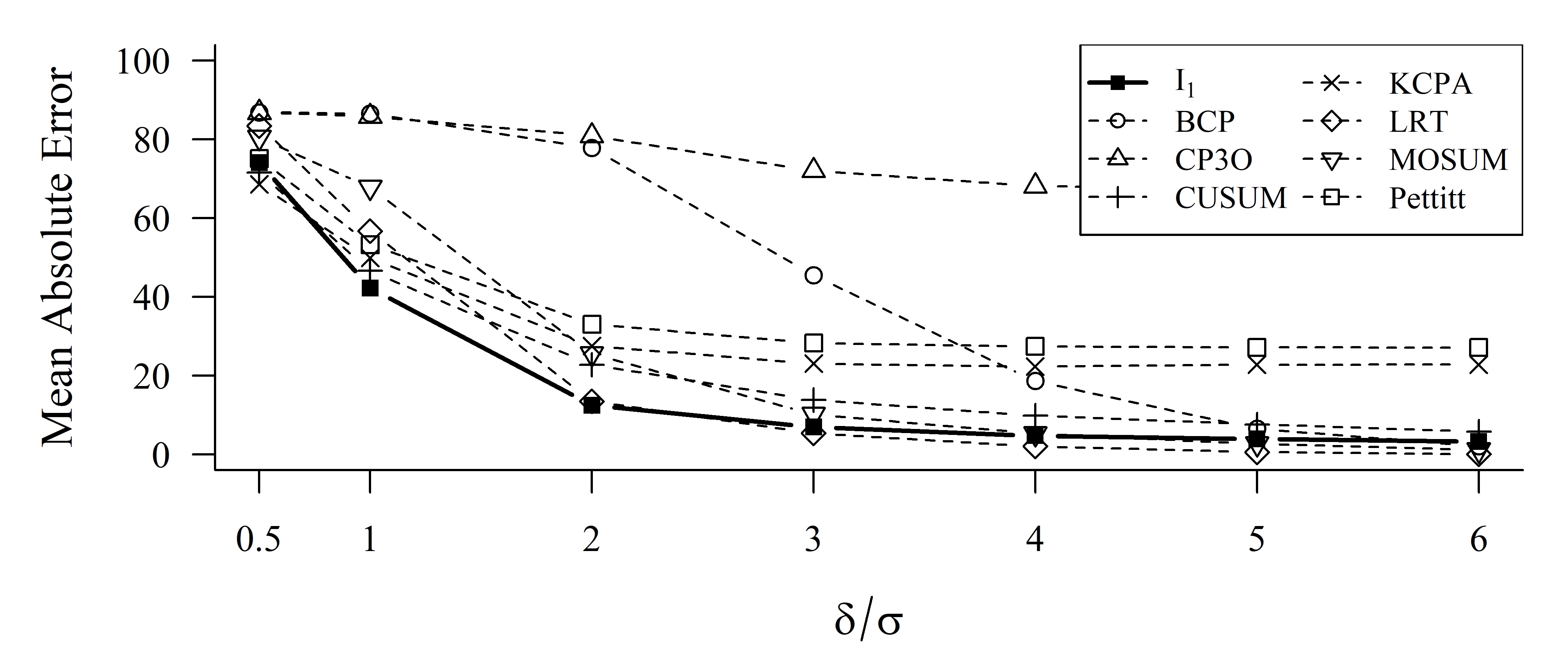}
        \caption{MAE vs $\frac{\delta}{\sigma}$ ($\rho<0$)}
        \label{fig:i1maemagarneg}
    \end{subfigure}
    
    \vspace{1em} 

    \begin{subfigure}[b]{0.32\textwidth}
        \centering
        \includegraphics[width=\textwidth]{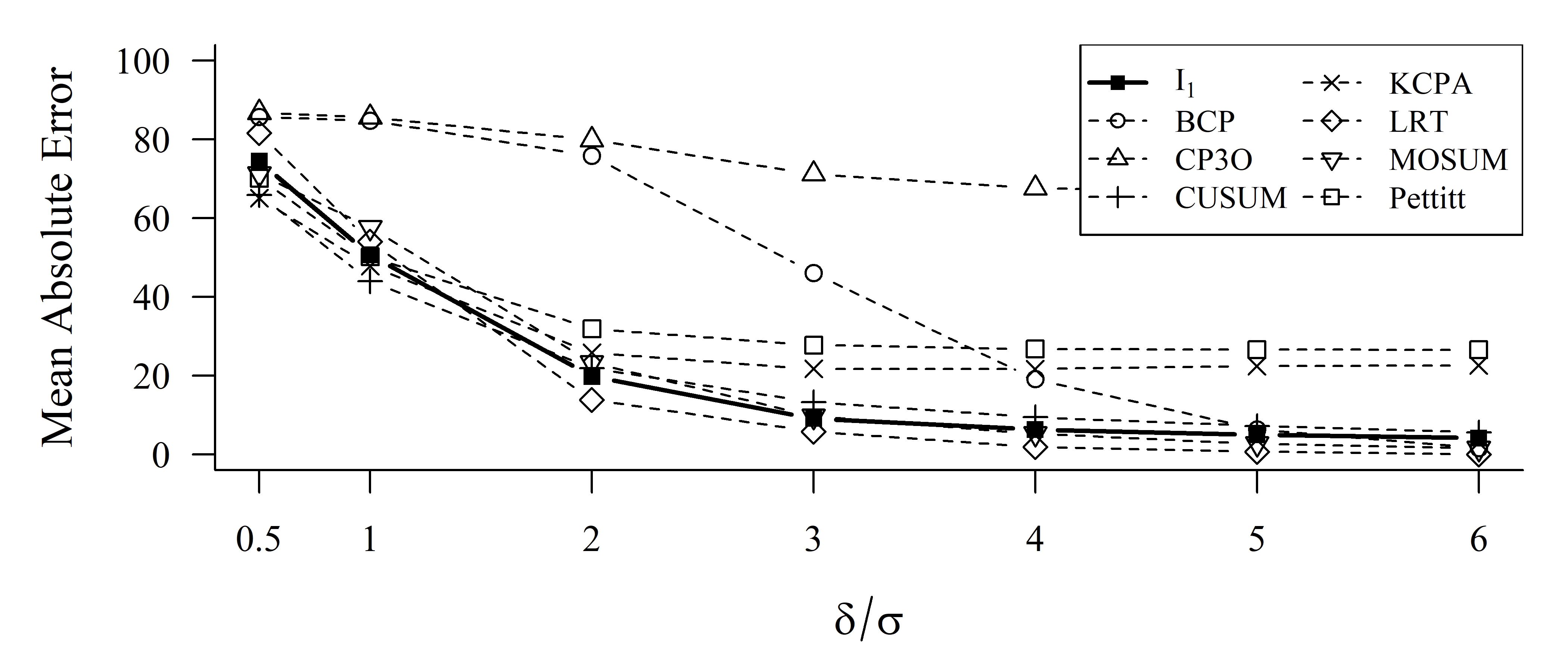}
        \caption{MAE vs $\frac{\delta}{\sigma}$ ($\rho>0$)}
        \label{fig:i1maemagarpoz}
    \end{subfigure}
    \begin{subfigure}[b]{0.32\textwidth}
        \centering
        \includegraphics[width=\textwidth]{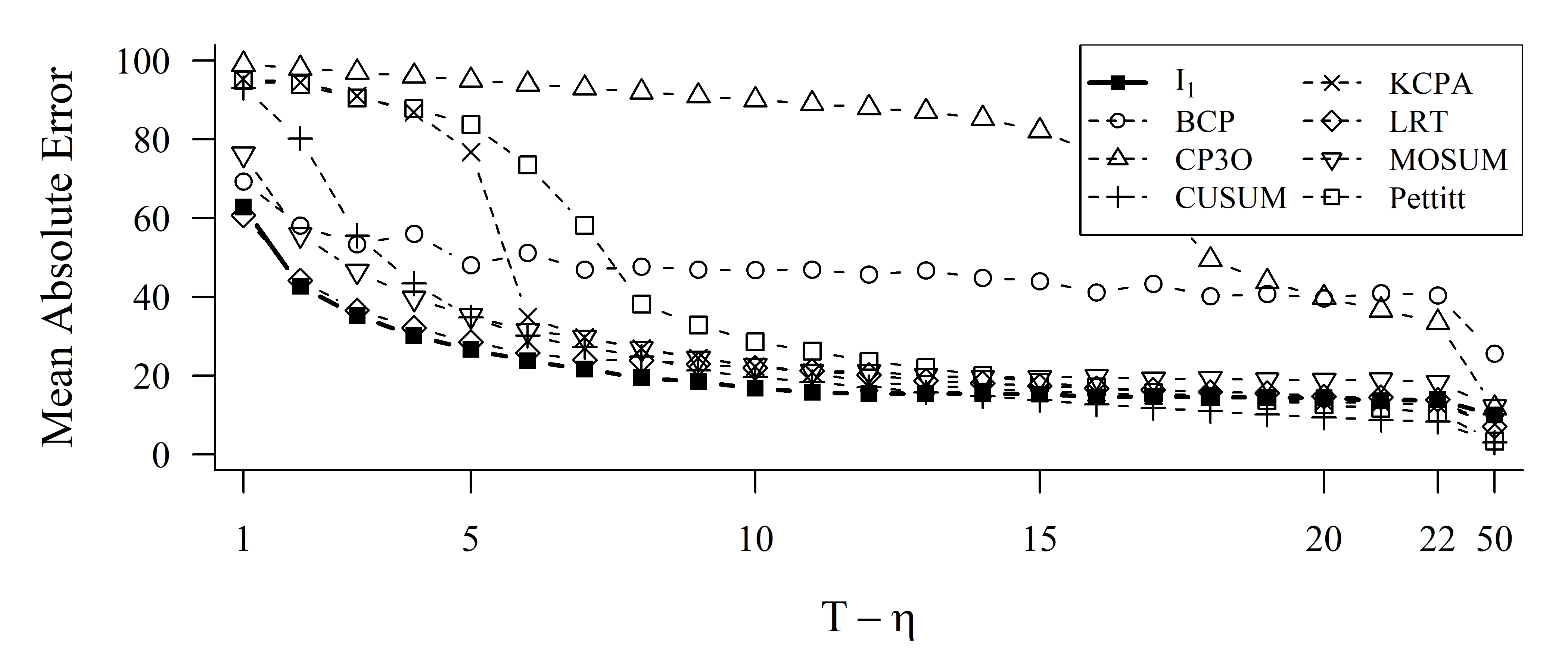}
        \caption{MAE vs $T-\eta$ ($\rho<0$)}
        \label{fig:i1maeetaarneg}
    \end{subfigure}
    \begin{subfigure}[b]{0.32\textwidth}
        \centering
        \includegraphics[width=\textwidth]{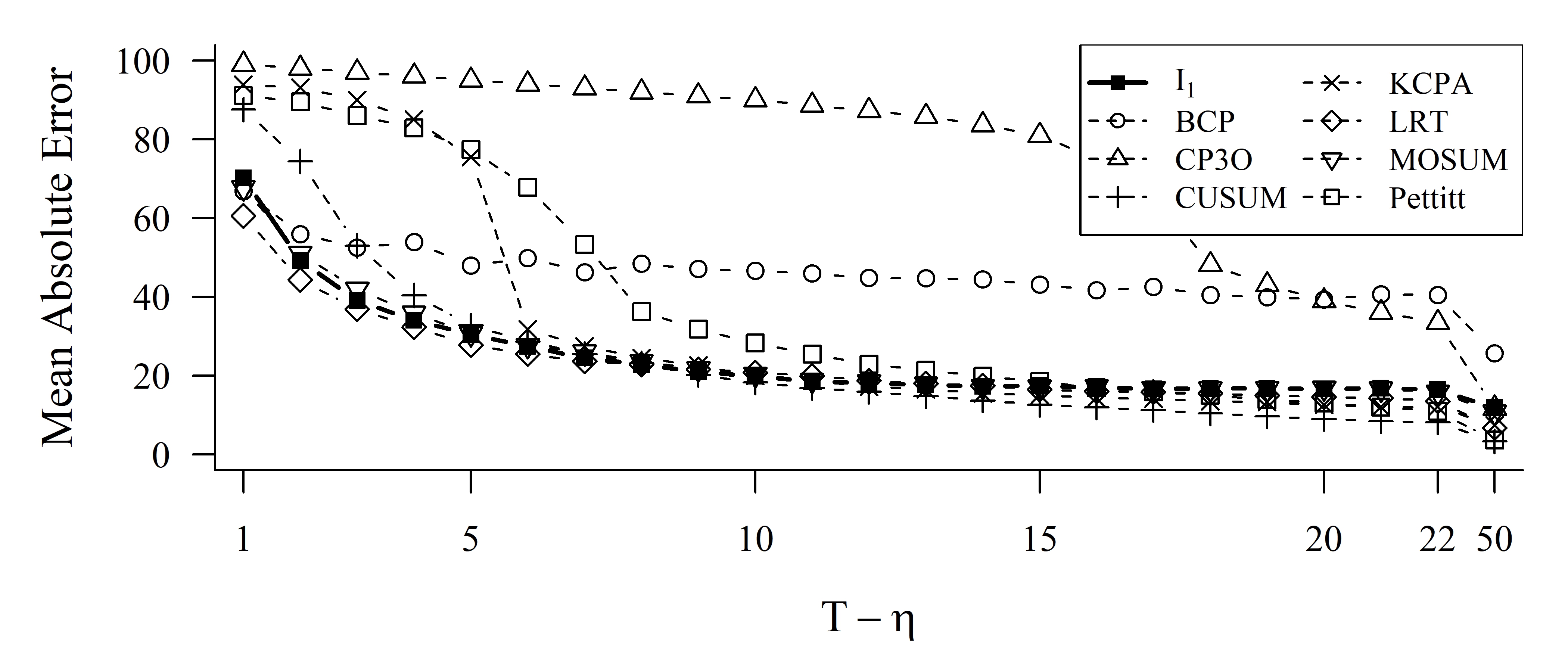}
        \caption{MAE vs $T-\eta$ ($\rho>0$)}
        \label{fig:i1maeetaarpoz}
    \end{subfigure}

    \caption{Performance of $I_1$ and benchmark tests across negative ($\rho<0$) and positive ($\rho>0$) autocorrelated datasets}
    \label{fig:overall_grid}
\end{figure}
\subsection{Performance Evaluation of DCI}
\label{ssec:pfmdci}
The performance of DCI is assessed by applying it, along with eight existing algorithms from the literature, to a diverse simulated dataset. Each simulated series embeds four groups of change points, well separated from one another and containing one, two, three, and four change points, respectively, with successive change points within a group separated by $\Delta t = \eta_{k+1} - \eta_k$. The temporal separation between groups is set to 30 time points. The experiment varies $\Delta t$ over $\{1, 2, 3, \ldots, 22, 50\}$, the change magnitude over $\frac{\delta}{\sigma} \in \{0.5, 1, 2, 3, 4, 5, 6\}$, and the data-generating process across three scenarios. Scenario 1 simulates time series sequences of i.i.d. observations. Scenario 2 simulates sequences of i.i.d. observations with a linear trend whose slope is drawn uniformly from $[-0.4,-0.1]\cup[0.1,0.4]$. Scenario 3 simulates a weakly dependent autoregressive sequences with $\rho$ uniformly drawn from $[-0.2, 0) \cup (0, 0.2]$.
This design confronts each algorithm simultaneously with isolated and closely-spaced change points within the same series, so that global and local detection abilities are tested jointly. By further crossing varying $\Delta t$ values with a range of $\frac{\delta}{\sigma}$ and 3 distinct data-generating scenarios, the setup provides a demanding and robust basis for evaluating multiple change-point detection algorithms. The experiment is conducted with $\alpha=0.1$.

The evaluation criteria have been slightly tightened to ensure a valid assessment of multiple change-point detection capabilities. The TPR is defined as before, but a detection qualifies as a true positive only if it is not unacceptably far from the actual change point. FPT is the average number of False Positives per Time series; any detected change points that do not qualify as true positives contribute to this metric. MAE is the mean of absolute deviations between the actual change points and the corresponding nearest detected change point(s). Finally, USCM denotes the Unweighted Segmentation Covering Metric. While the segmentation-covering metric is known for evaluating segmentation \parencite{van2020evaluation}, it weights performance by the length of individual segments, thereby down-weighting exactly those segments that exhibit rapid multiple change points and that we wish to detect with priority. Let $\mathcal{G}$ and $\mathcal{G}'$ denote the actual and detected partitions, respectively, each comprising mutually exclusive and collectively exhaustive segments $\mathcal{A}\in\mathcal{G}$ and $\mathcal{A}'\in\mathcal{G}'$. $
USCM(\mathcal{G},\mathcal{G}')
=
\frac{1}{|\mathcal{G}|}
\sum_{\mathcal{A}\in\mathcal{G}}
\max_{\mathcal{A}'\in\mathcal{G}'}
\frac{|\mathcal{A} \cap \mathcal{A}'|}{|\mathcal{A} \cup \mathcal{A}'|}
$ measures covering performance by assigning equal importance to each actual segment.

The benchmark algorithms are as follows. The function \texttt{cpt.mean} from the R package \texttt{changepoint} \parencite{killick2014changepoint} deploys the binary segmentation (BinSeg) algorithm of \textcite{scott1974cluster}. We further include the narrowest-over-threshold (NOT) method \parencite{baranowski2019narrowest}, EnvCpt \parencite{beaulieu2018distinguishing}, continuous-piecewise linear pruned optimal partitioning (CPOP) \parencite{fearnhead2019detecting}, and TrendSegment \parencite{maeng2024detecting}. The NOT algorithm is implemented from the function \texttt{not} in the R package \texttt{not} \parencite{baranowski2019not}, and the EnvCpt algorithm, which deploys the PELT algorithm, from the function \texttt{envcpt} in the R package \texttt{EnvCpt} \parencite{killick2021envcpt}.

Table~\ref{tab:pfmdciiid} reports the aggregate performance of the default DCI and DCI$_{WD}$ (which employs the `WD' variant of $I_1$ for weakly dependent data) alongside the benchmark algorithms. A more detailed evaluation is provided in Figures~\ref{fig:performance-timegap} and \ref{fig:performance-delta-sigma}, which trace TPR, FPT, and MAE across the simulation settings. Across these diverse settings, DCI remains Pareto-optimal and demonstrates the strongest dominance over benchmark methods in pairwise criteria spaces. The Multi-Objective Value (MOV) of each method can be computed as the weighted sum of the \(L^2\)-normalized criterion values \parencite{waychal2023adaptive}. Assuming equal weights across criteria in the TPR–FPT–MAE–USCM space (Table~\ref{tab:pfmdciiid}), the two DCI variants attain the highest MOVs.

\begin{table*}
\centering
\begin{threeparttable}

\setlength{\tabcolsep}{5.5pt}
\renewcommand{\arraystretch}{1.15}

\begin{tabular}{
    l
    S[table-format=1.3]
    S[table-format=1.2]
    S[table-format=2.2]
    S[table-format=1.3]
    S[table-format=1.2]
    S[table-format=2.0]
    S[table-format=2.0]
}
\toprule

& & & & & &
\multicolumn{2}{c}{\makecell{Number of alternatives dominated}} \\

\cmidrule(lr){7-8}

\makecell[l]{Detection\\method}
&
{\makecell{TPR\\$\uparrow$}}
&
{\makecell{FPT\\$\downarrow$}}
&
{\makecell{MAE\\$\downarrow$}}
&
{\makecell{USCM\\$\uparrow$}}
&
{\makecell{MOV\\$\uparrow$}}
&
{\makecell{TPR-FPT\\space}}
&
{\makecell{TPR-MAE\\space}}
\\

\midrule

DCI           & 0.614 & 1.22 & 13.12 & 0.522 &  0.11 & 3 &  9 \\
DCI$_{WD}$    & 0.569 & 1.13 & 15.14 & 0.495 &  0.09 & 1 &  6 \\
BCP           & 0.165 & 0.06 & 63.22 & 0.241 & -0.08 & 0 &  0 \\
BinSeg        & 0.597 & 1.84 & 17.26 & 0.535 &  0.07 & 1 &  7 \\
CP3O          & 0.405 & 0.63 & 12.81 & 0.428 &  0.08 & 2 &  6 \\
CPOP          & 0.337 & 1.89 & 29.60 & 0.349 & -0.03 & 0 &  4 \\
KCPA          & 0.749 & 5.16 &  6.43 & 0.600 &  0.00 & 0 & 11 \\
MOSUM         & 0.571 & 1.65 & 23.70 & 0.525 &  0.06 & 1 &  6 \\
NOT           & 0.319 & 0.19 & 31.19 & 0.367 &  0.04 & 0 &  3 \\
PELT          & 0.219 & 0.09 & 40.99 & 0.287 & -0.01 & 0 &  2 \\
Pettitt       & 0.095 & 0.05 & 54.42 & 0.178 & -0.08 & 0 &  0 \\
TrSeg         & 0.358 & 0.63 & 28.37 & 0.386 &  0.03 & 1 &  5 \\

\bottomrule
\end{tabular}

\begin{tablenotes}[flushleft]
\footnotesize
\item \textit{Note:}
$\uparrow$ ($\downarrow$) indicates that larger (smaller) values are preferred.
\end{tablenotes}

\end{threeparttable}
\caption{Performance comparison of multiple change-points detection methods}
\label{tab:pfmdciiid}
\end{table*}
\begin{figure*}
    \centering

    \begin{subfigure}[t]{0.325\textwidth}
        \centering
        \includegraphics[width=\linewidth]{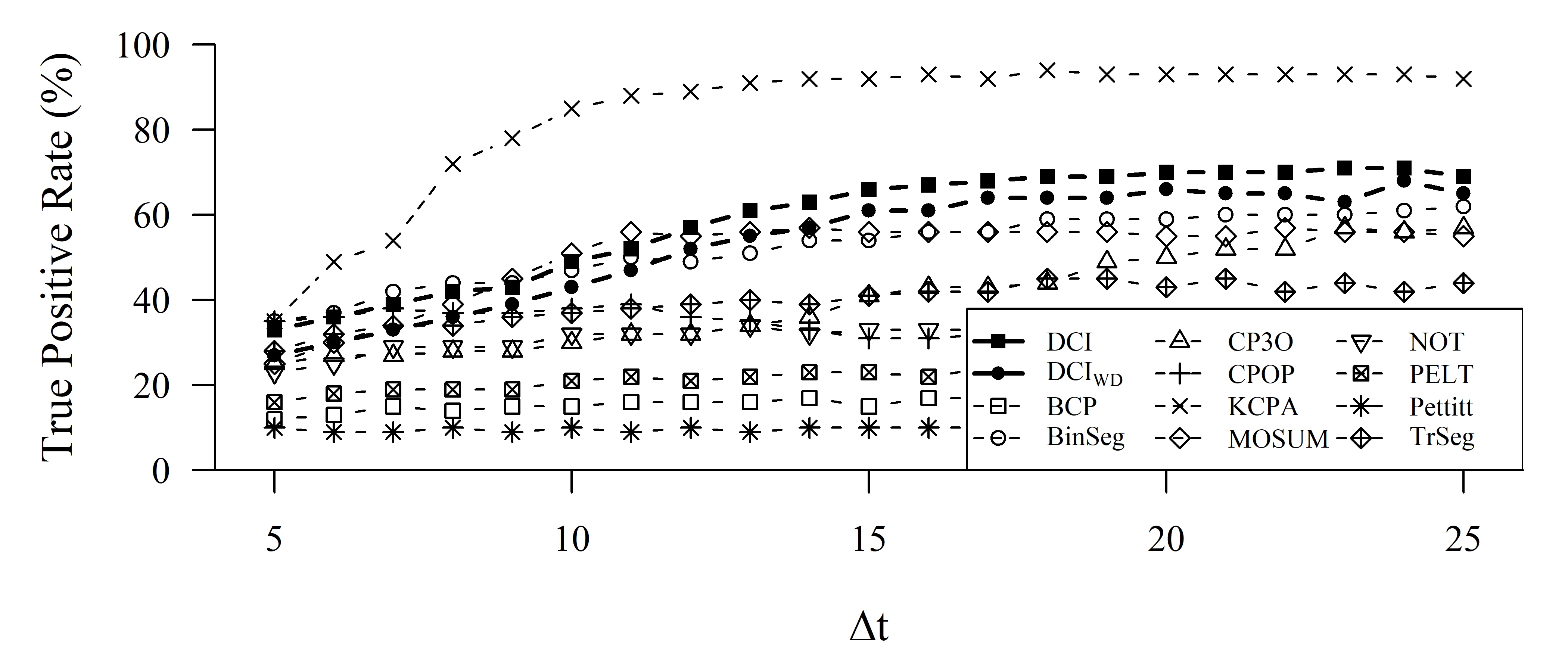}
        \caption{TPR vs $\Delta t$ (Scenario 1)}
        \label{fig:timegaps1-tpr}
    \end{subfigure}
    \begin{subfigure}[t]{0.325\textwidth}
        \centering
        \includegraphics[width=\linewidth]{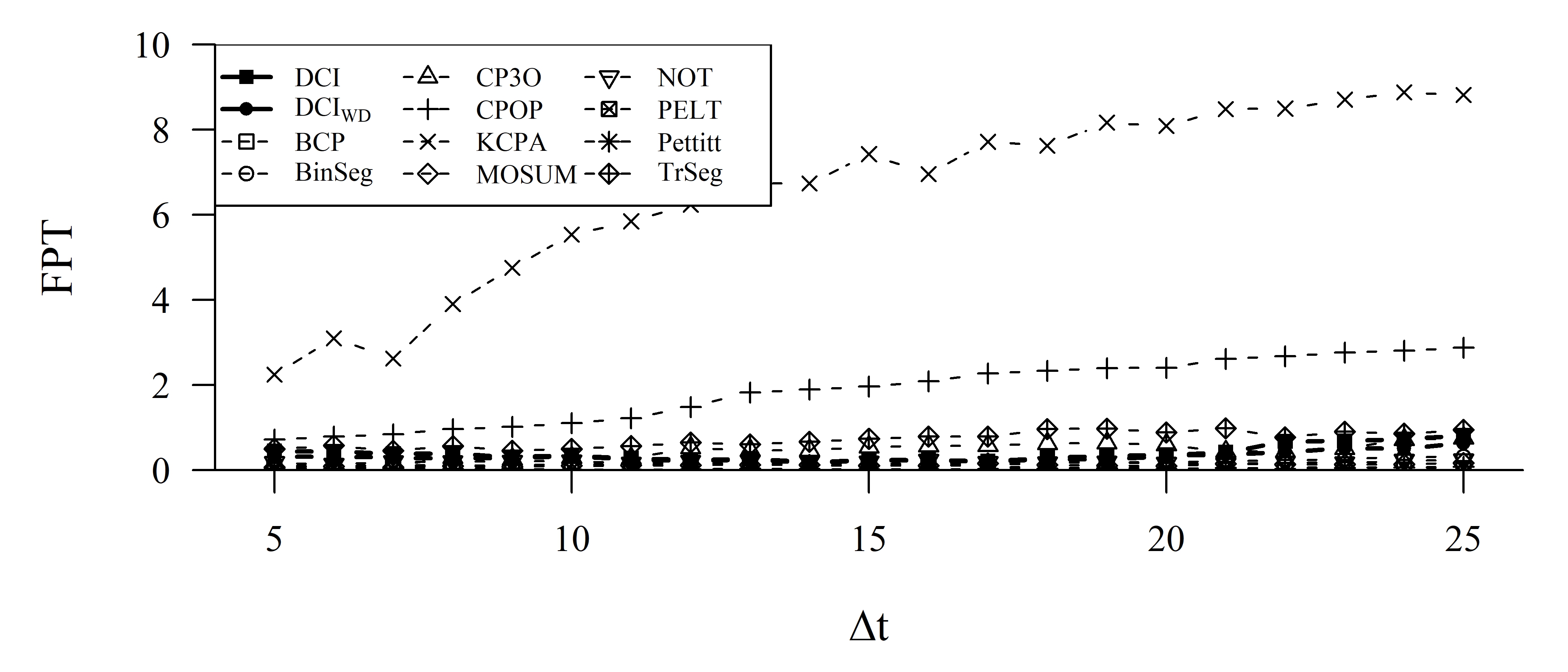}
        \caption{FPT vs $\Delta t$ (Scenario 1)}
        \label{fig:timegaps1-fpt}
    \end{subfigure}
    \begin{subfigure}[t]{0.325\textwidth}
        \centering
        \includegraphics[width=\linewidth]{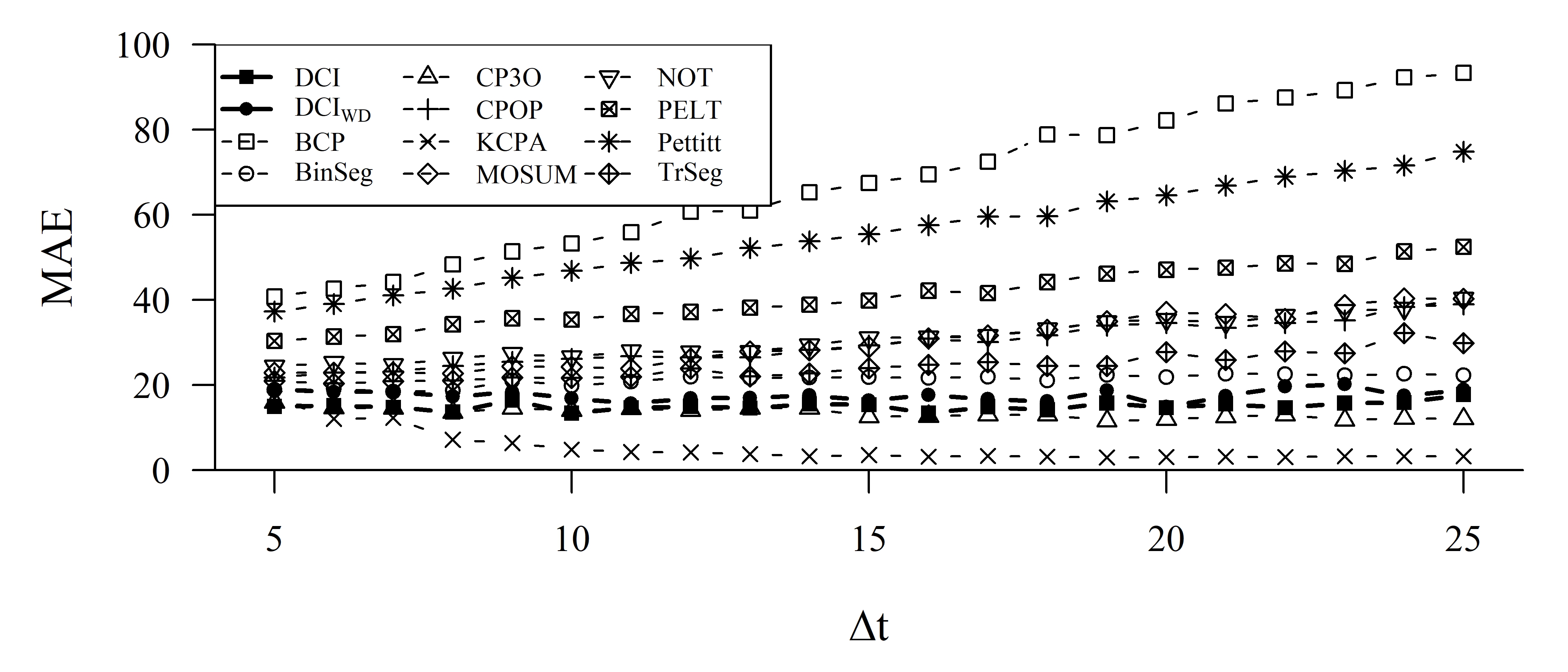}
        \caption{MAE vs $\Delta t$ (Scenario 1)}
        \label{fig:timegaps1-mae}
    \end{subfigure}

    \vspace{0.5em}

    \begin{subfigure}[t]{0.325\textwidth}
        \centering
        \includegraphics[width=\linewidth]{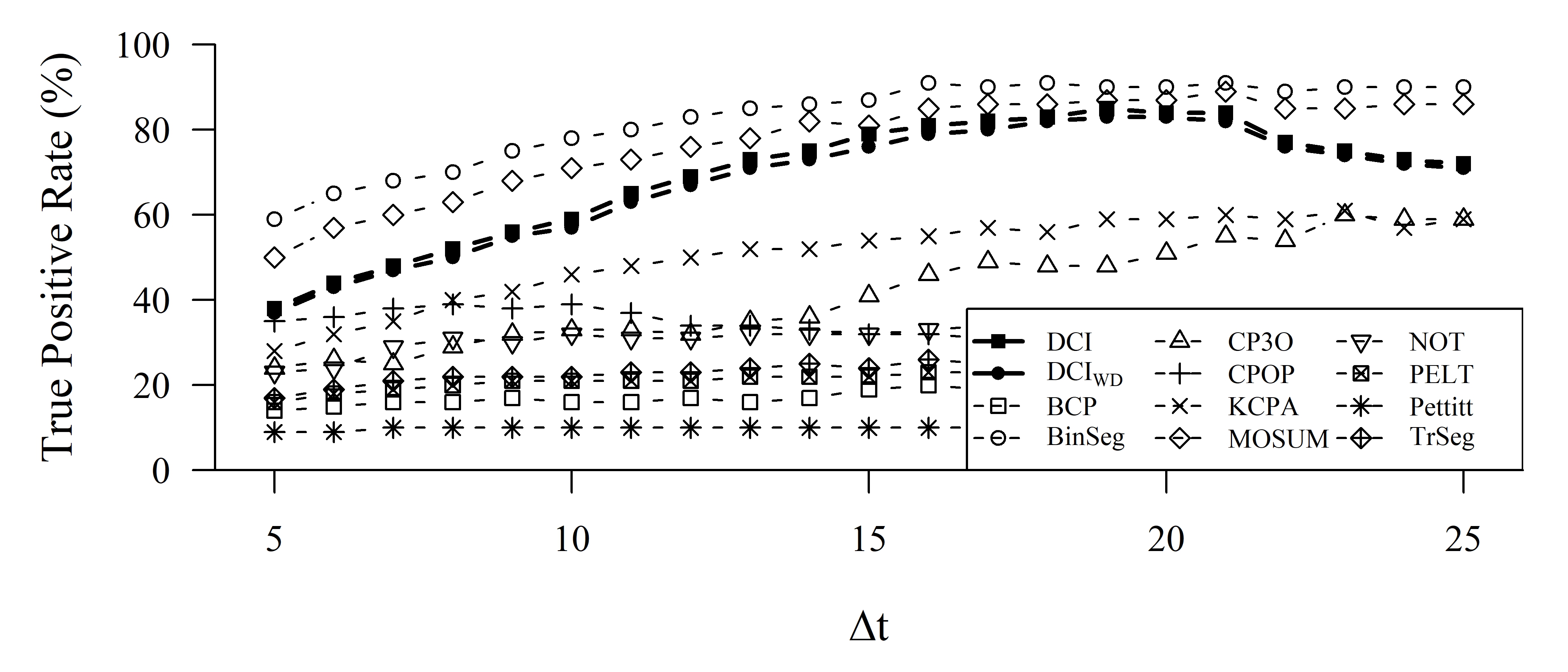}
        \caption{TPR vs $\Delta t$ (Scenario 2)}
        \label{fig:timegaps2-tpr}
    \end{subfigure}
    \begin{subfigure}[t]{0.325\textwidth}
        \centering
        \includegraphics[width=\linewidth]{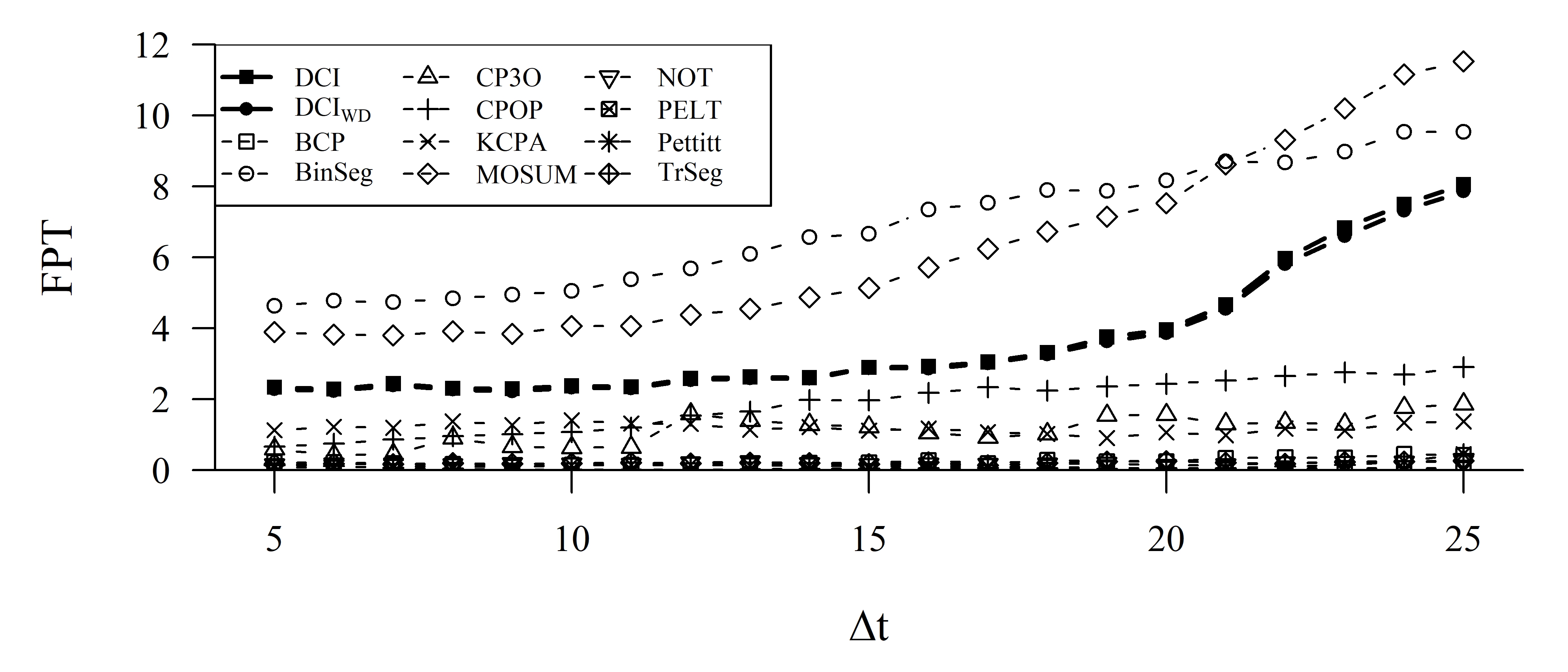}
        \caption{FPT vs $\Delta t$ (Scenario 2)}
        \label{fig:timegaps2-fpt}
    \end{subfigure}
    \begin{subfigure}[t]{0.325\textwidth}
        \centering
        \includegraphics[width=\linewidth]{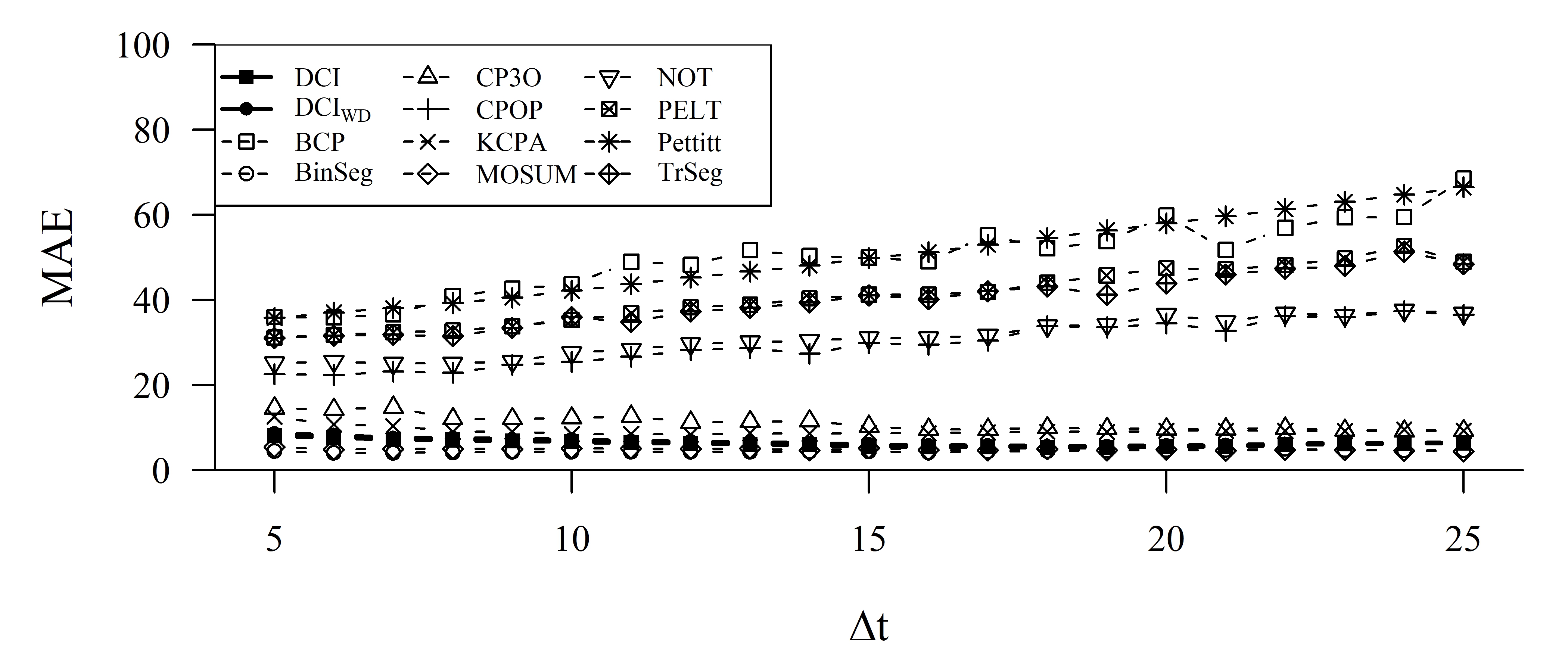}
        \caption{MAE vs $\Delta t$ (Scenario 2)}
        \label{fig:timegaps2-mae}
    \end{subfigure}

    \vspace{0.5em}

    \begin{subfigure}[t]{0.325\textwidth}
        \centering
        \includegraphics[width=\linewidth]{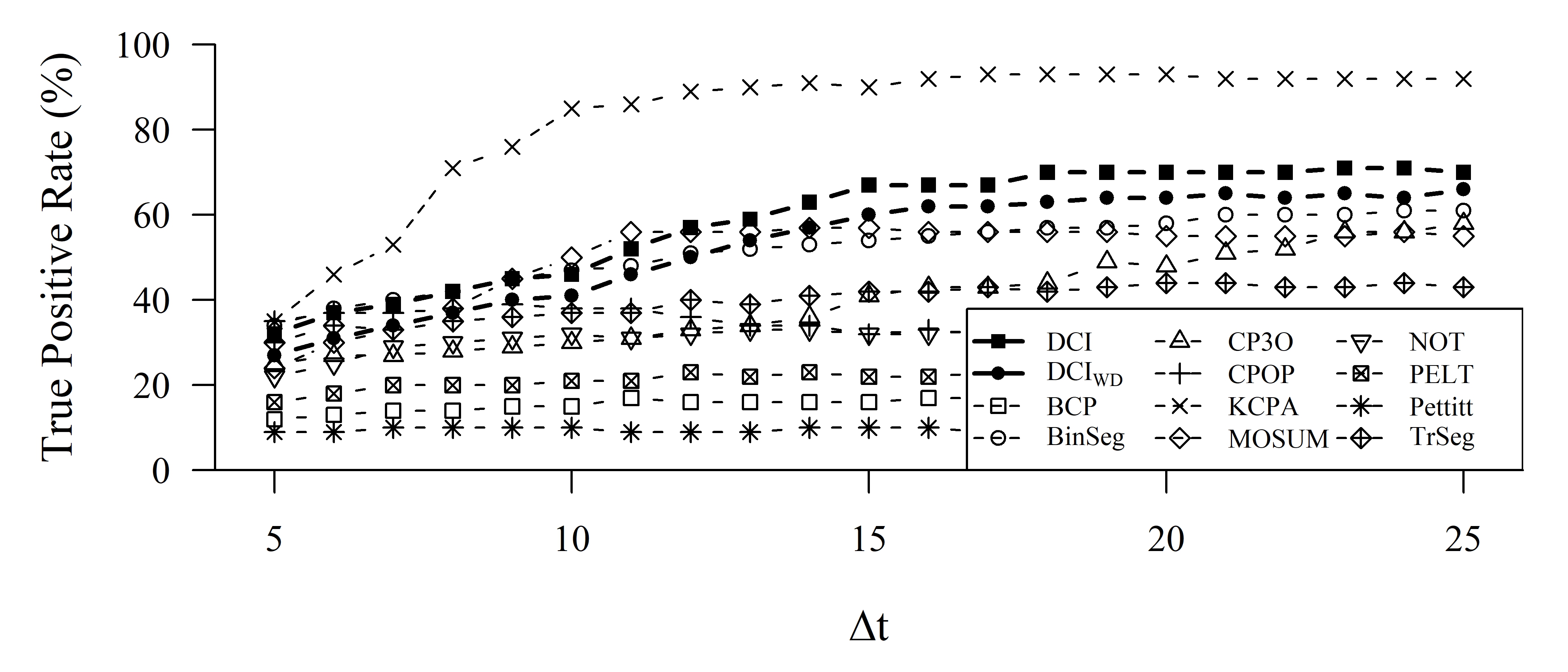}
        \caption{TPR vs $\Delta t$ (Scenario 3)}
        \label{fig:timegaps3-tpr}
    \end{subfigure}
    \begin{subfigure}[t]{0.325\textwidth}
        \centering
        \includegraphics[width=\linewidth]{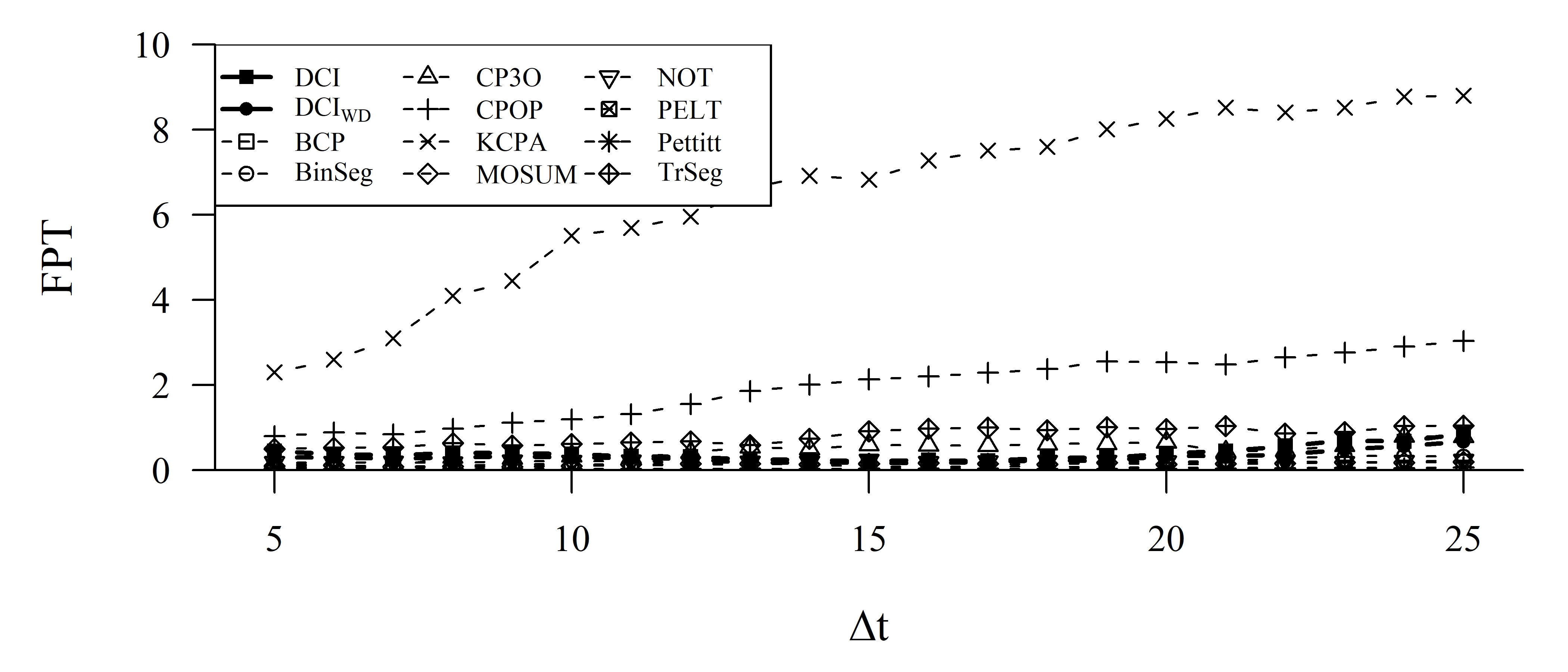}
        \caption{FPT vs $\Delta t$ (Scenario 3)}
        \label{fig:timegaps3-fpt}
    \end{subfigure}
    \begin{subfigure}[t]{0.325\textwidth}
        \centering
        \includegraphics[width=\linewidth]{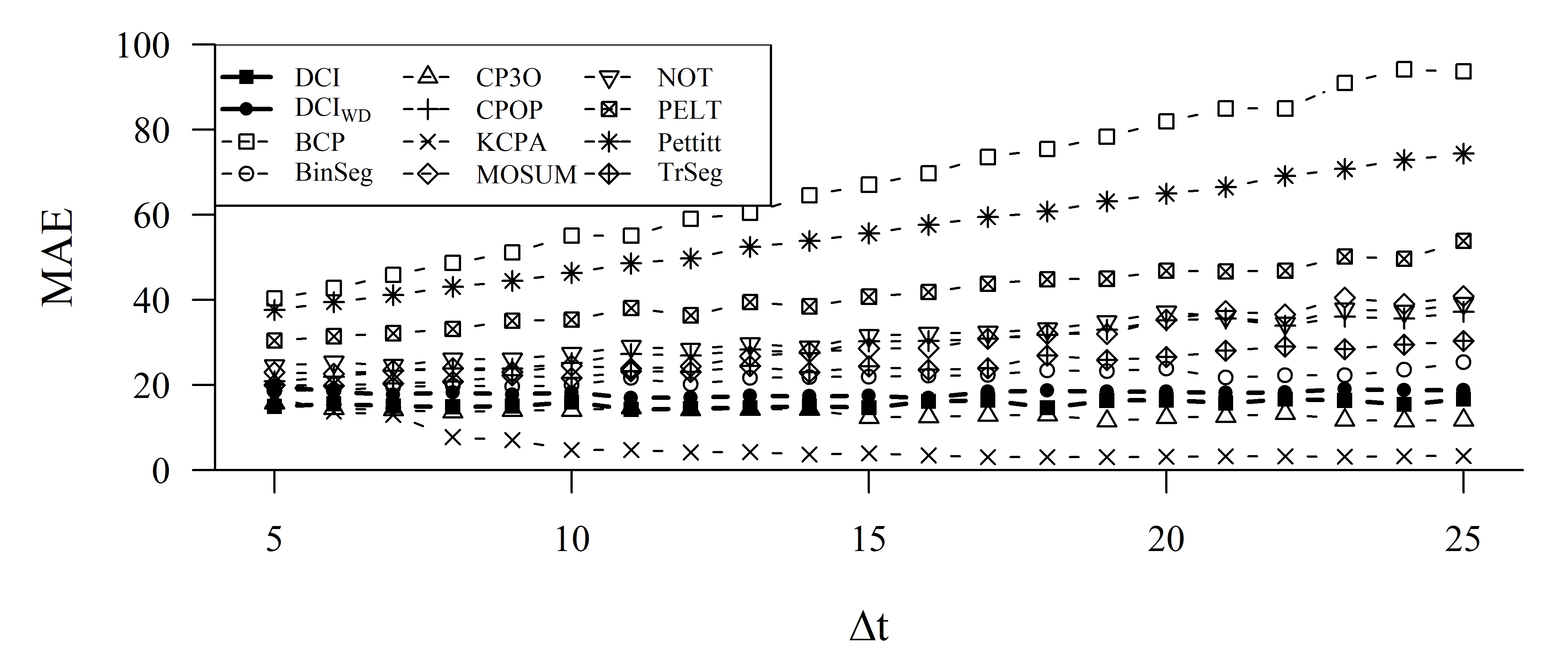}
        \caption{MAE vs $\Delta t$ (Scenario 3)}
        \label{fig:timegaps3-mae}
    \end{subfigure}

    \caption{Performance of the change-point detection methods as a function
    of $\Delta t$ under the three
    simulation scenarios.}
    \label{fig:performance-timegap}
\end{figure*}
\begin{figure*}
    \centering

    \begin{subfigure}[t]{0.325\textwidth}
        \centering
        \includegraphics[width=\linewidth]{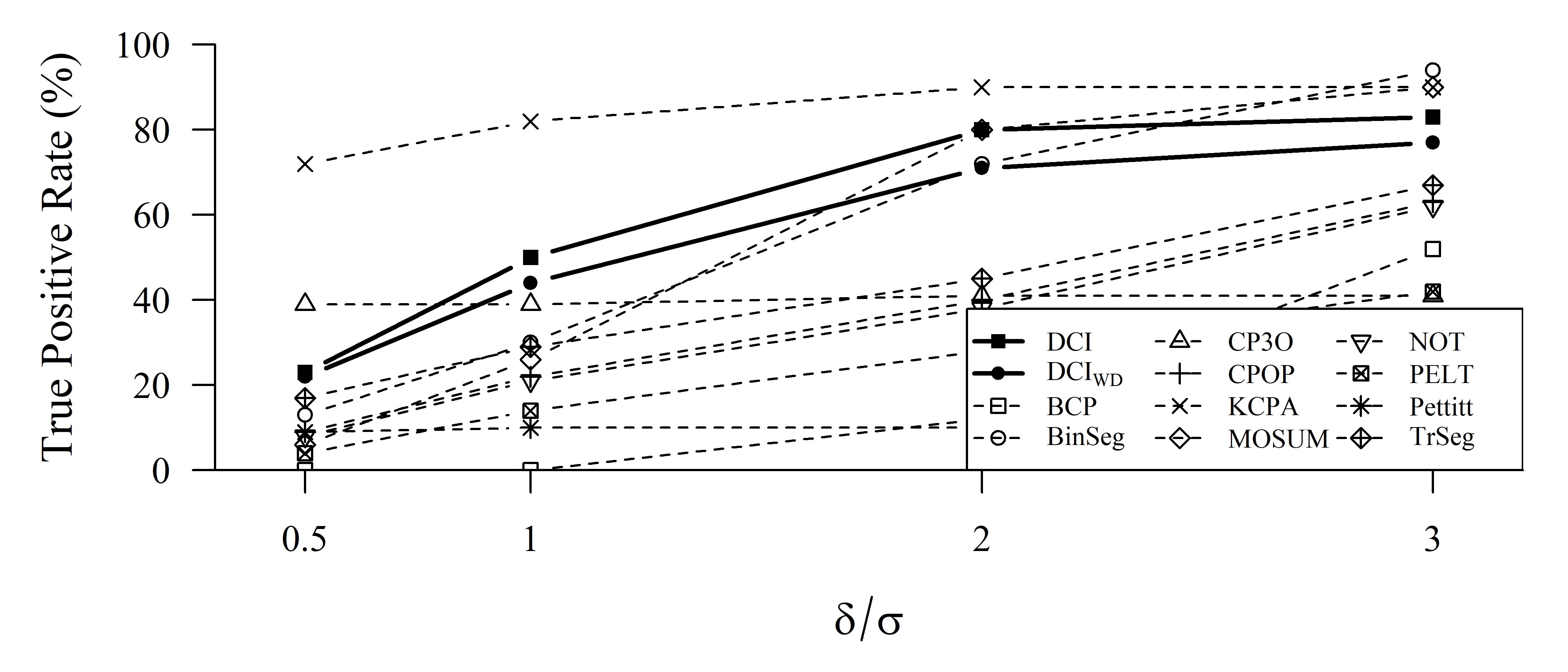}
        \caption{TPR vs $\frac{\delta}{\sigma}$ (Scenario 1)}
        \label{fig:dels-s1-tpr}
    \end{subfigure}
    \begin{subfigure}[t]{0.325\textwidth}
        \centering
        \includegraphics[width=\linewidth]{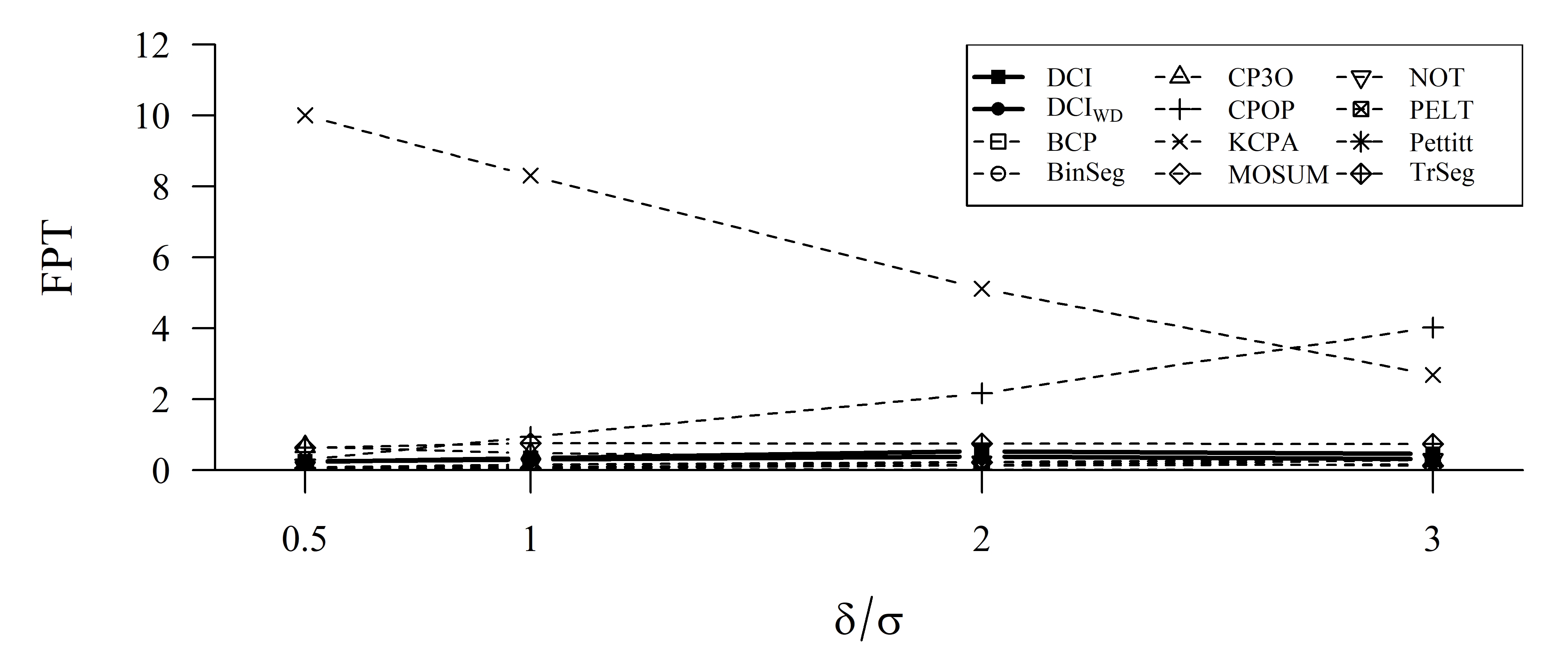}
        \caption{FPT vs $\frac{\delta}{\sigma}$ (Scenario 1)}
        \label{fig:dels-s1-fpt}
    \end{subfigure}
    \begin{subfigure}[t]{0.325\textwidth}
        \centering
        \includegraphics[width=\linewidth]{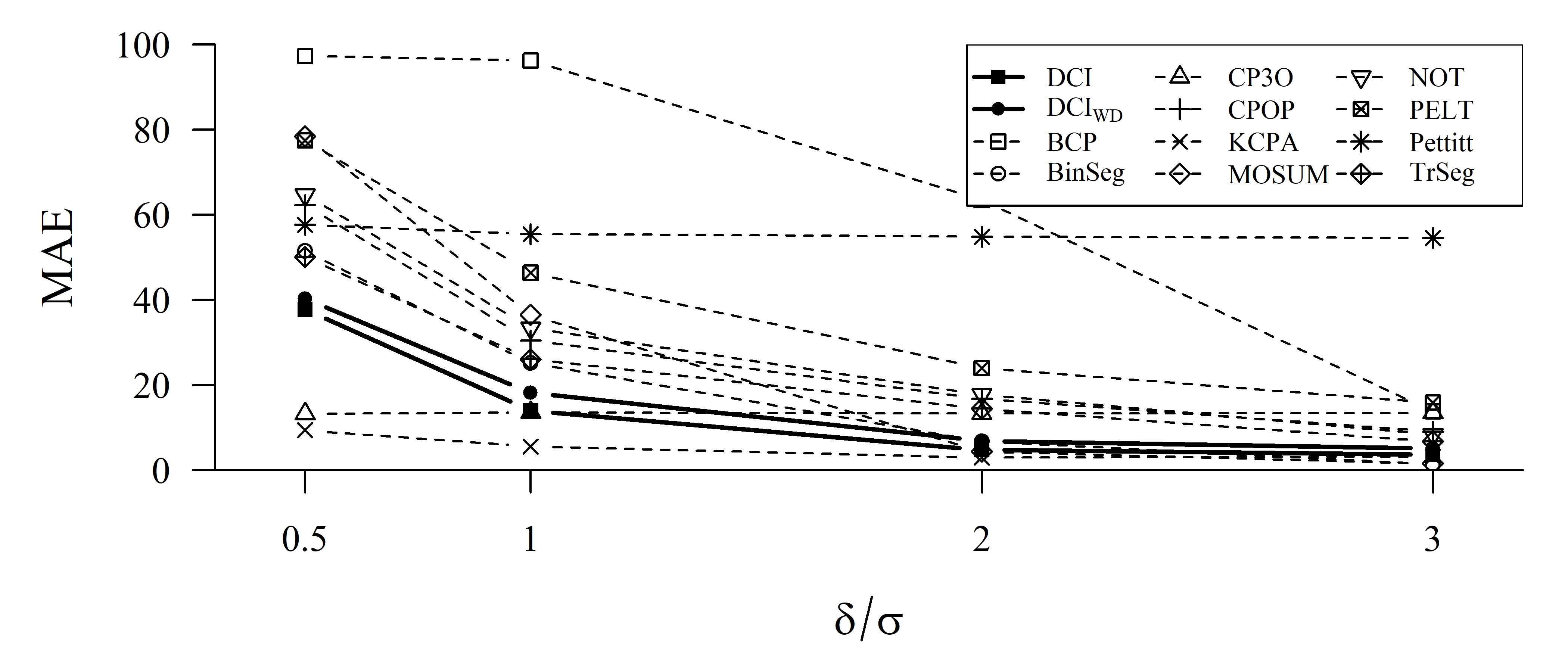}
        \caption{MAE vs $\frac{\delta}{\sigma}$ (Scenario 1)}
        \label{fig:dels-s1-mae}
    \end{subfigure}

    \vspace{0.5em}

    \begin{subfigure}[t]{0.325\textwidth}
        \centering
        \includegraphics[width=\linewidth]{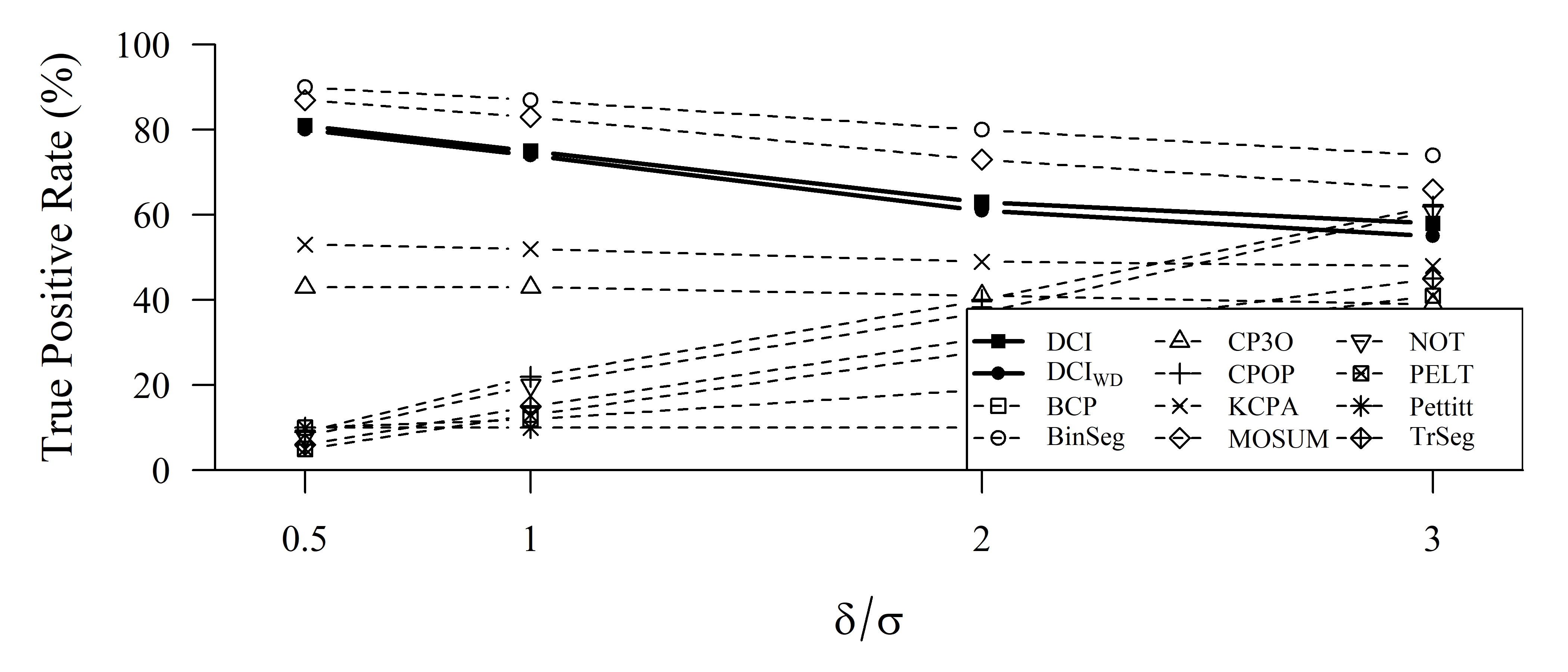}
        \caption{TPR vs $\frac{\delta}{\sigma}$ (Scenario 2)}
        \label{fig:dels-s2-tpr}
    \end{subfigure}
    \begin{subfigure}[t]{0.325\textwidth}
        \centering
        \includegraphics[width=\linewidth]{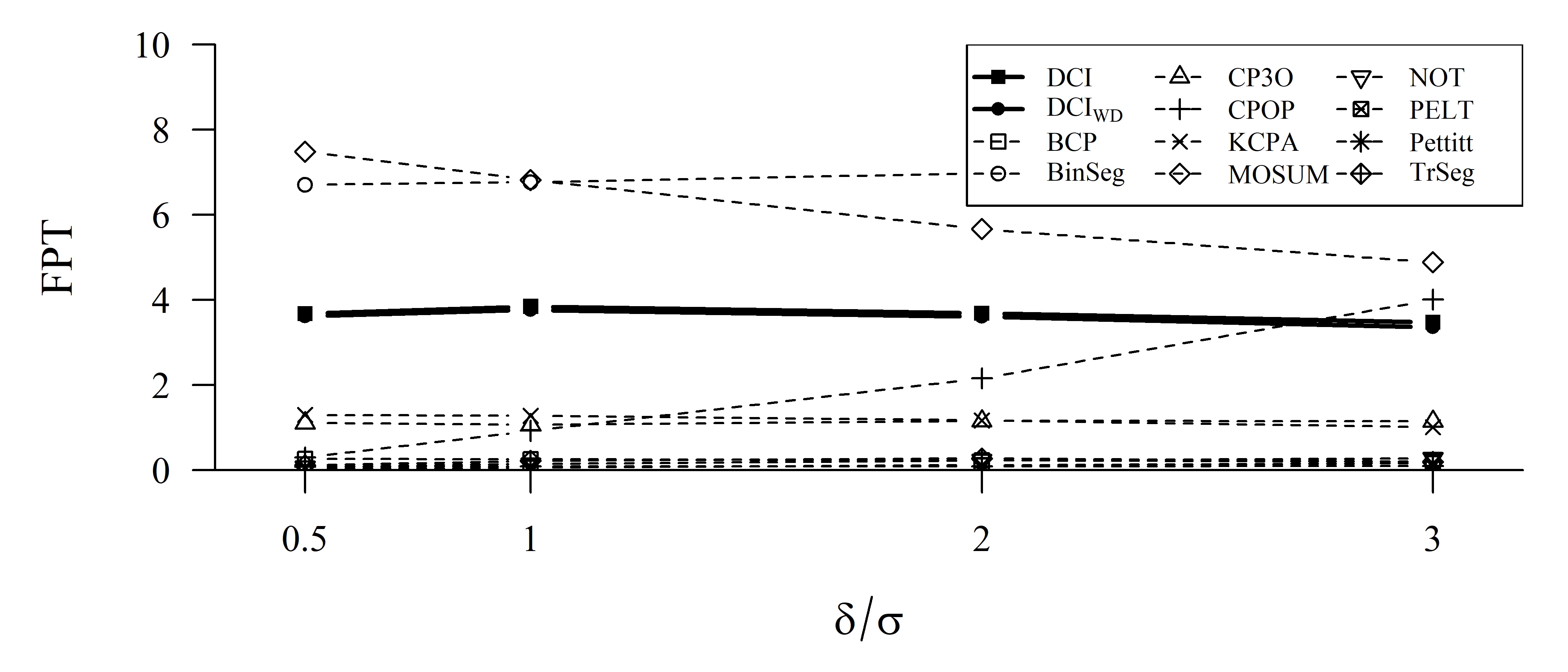}
        \caption{FPT vs $\frac{\delta}{\sigma}$ (Scenario 2)}
        \label{fig:dels-s2-fpt}
    \end{subfigure}
    \begin{subfigure}[t]{0.325\textwidth}
        \centering
        \includegraphics[width=\linewidth]{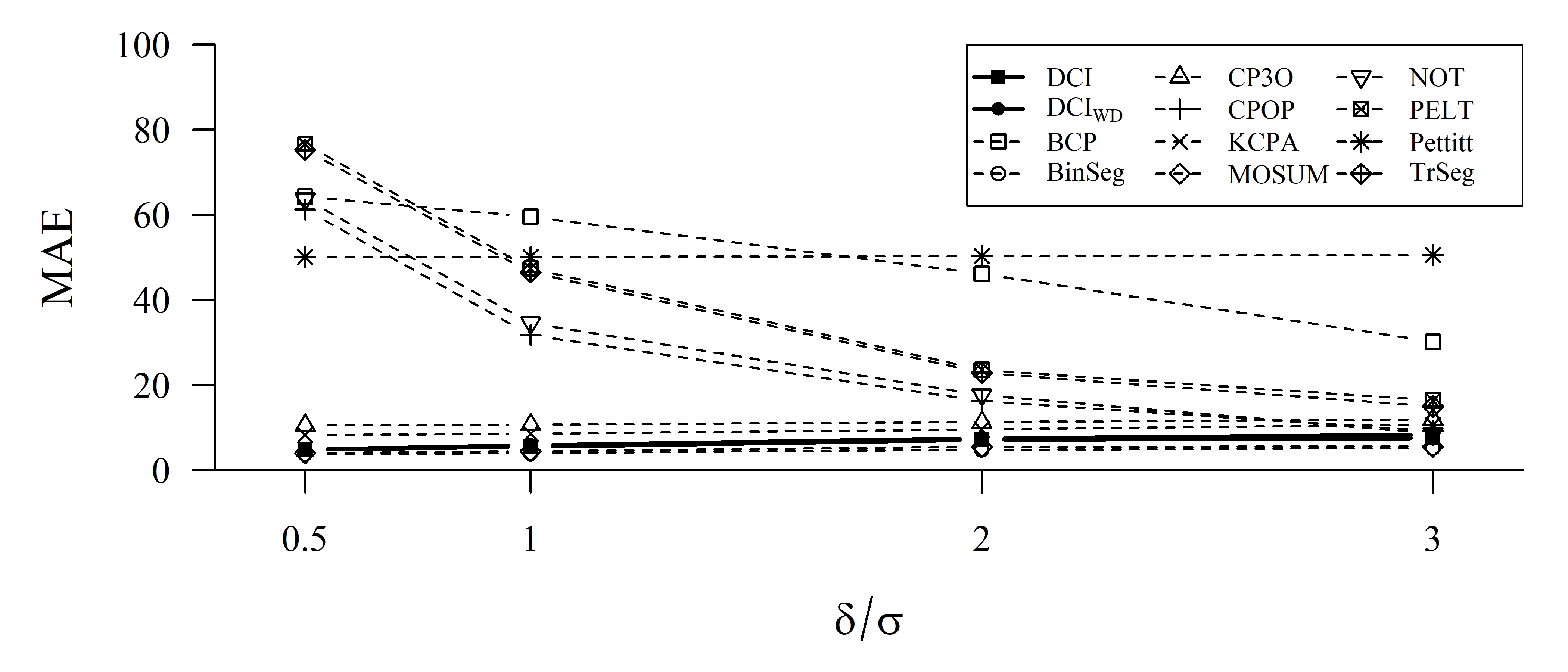}
        \caption{MAE vs $\frac{\delta}{\sigma}$ (Scenario 2)}
        \label{fig:dels-s2-mae}
    \end{subfigure}

    \vspace{0.5em}

    \begin{subfigure}[t]{0.325\textwidth}
        \centering
        \includegraphics[width=\linewidth]{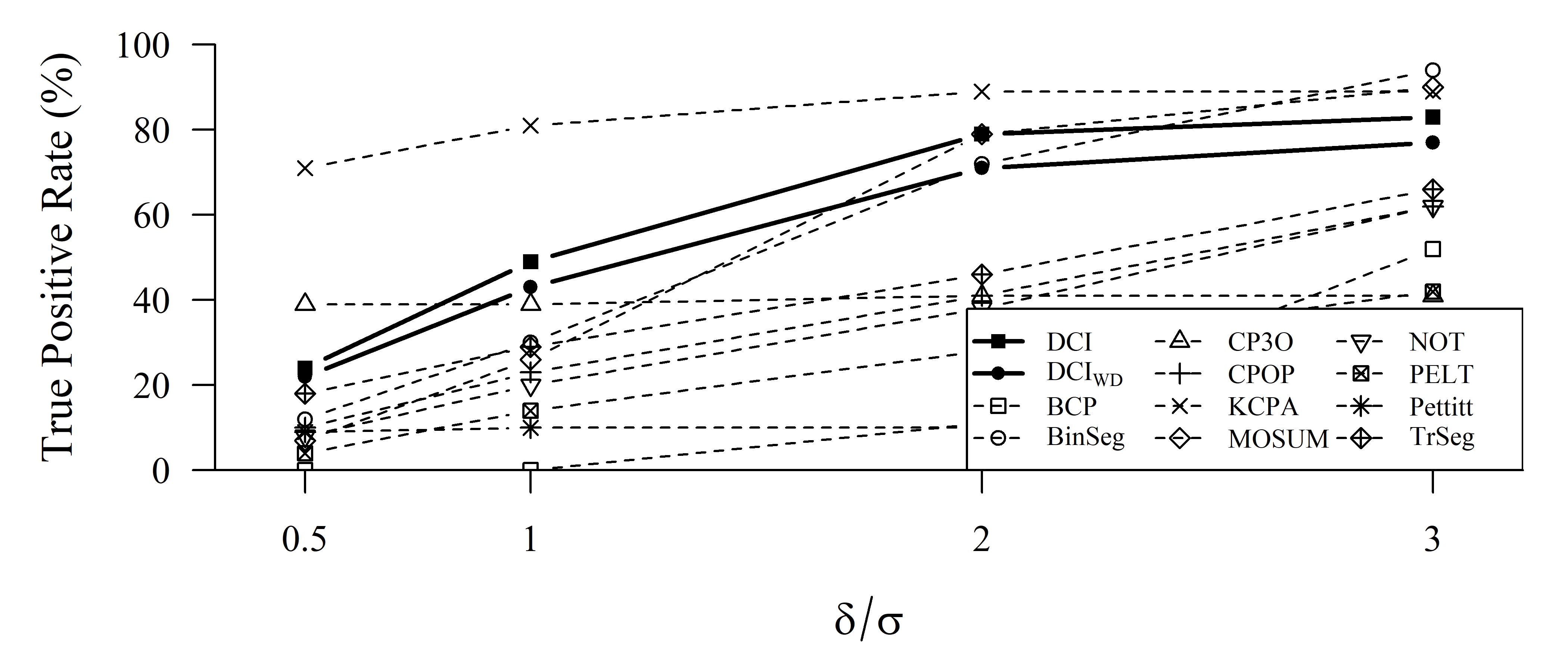}
        \caption{TPR vs $\frac{\delta}{\sigma}$ (Scenario 3)}
        \label{fig:dels-s3-tpr}
    \end{subfigure}
    \begin{subfigure}[t]{0.325\textwidth}
        \centering
        \includegraphics[width=\linewidth]{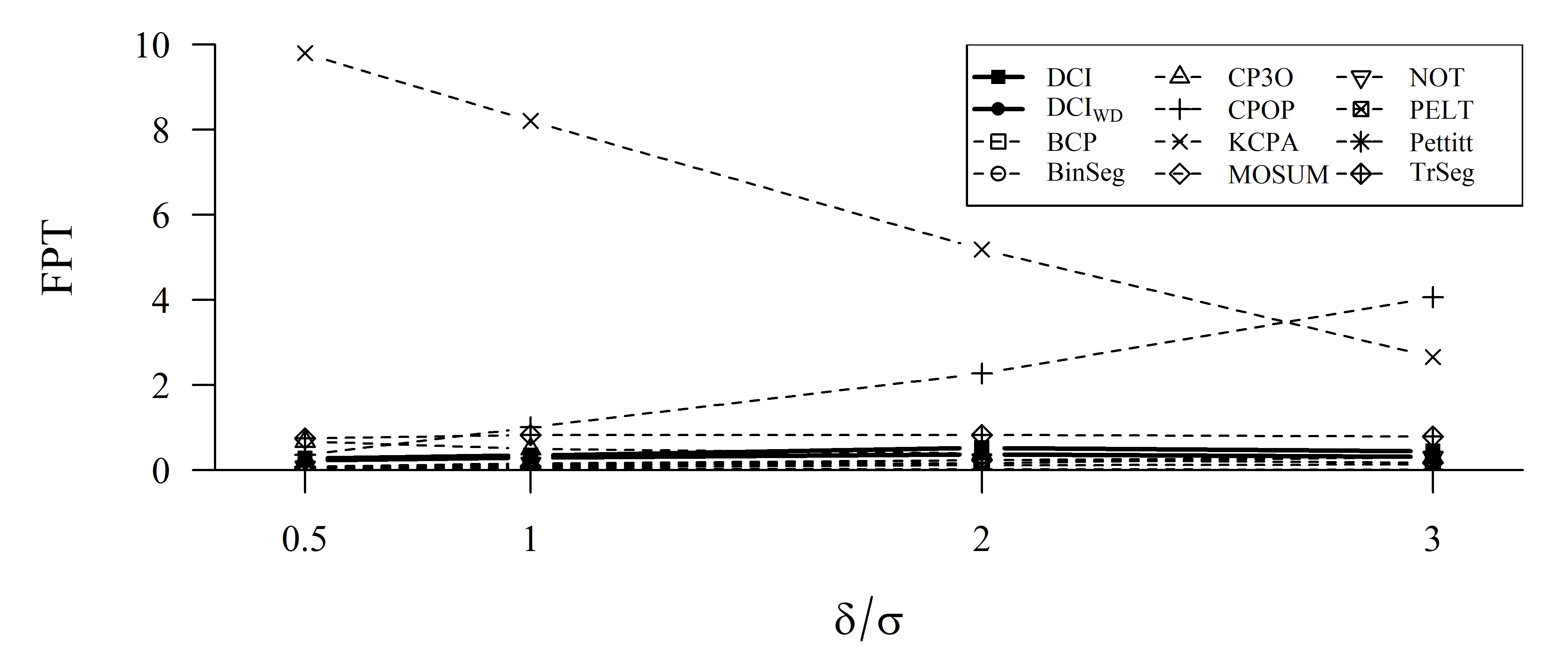}
        \caption{FPT vs $\frac{\delta}{\sigma}$ (Scenario 3)}
        \label{fig:dels-s3-fpt}
    \end{subfigure}
    \begin{subfigure}[t]{0.325\textwidth}
        \centering
        \includegraphics[width=\linewidth]{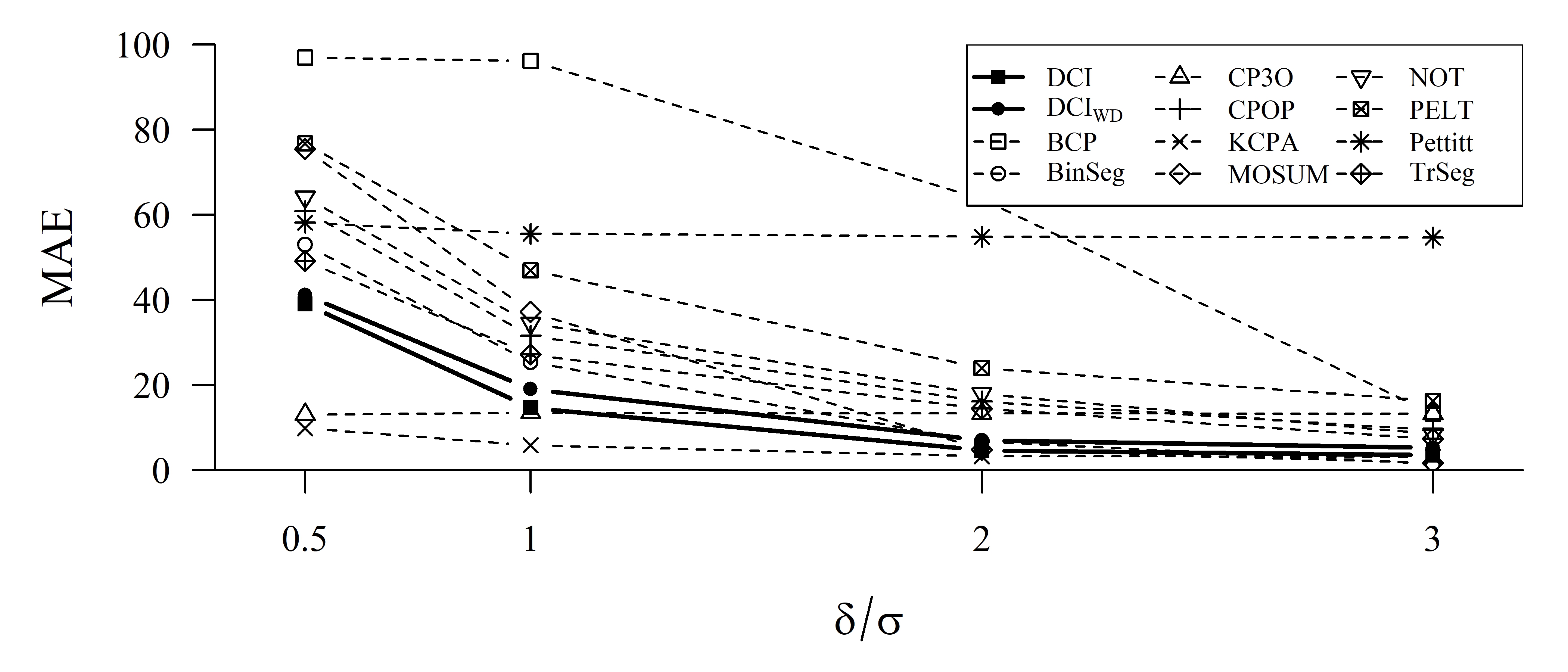}
        \caption{MAE vs $\frac{\delta}{\sigma}$ (Scenario 3)}
        \label{fig:dels-s3-mae}
    \end{subfigure}

    \caption{Performance of the change-point detection methods as a function
    of $\frac{\delta}{\sigma}$ under the three
    simulation scenarios.}
    \label{fig:performance-delta-sigma}
\end{figure*}
\section{Empirical Application: The US-Iran Geopolitical Conflict and Energy Market Dynamics}
\label{sec:app}
We now demonstrate the economic validity of the proposed framework in a high-stakes macroeconomic setting: the structural impact of US-Iran geopolitical tensions surrounding the Strait of Hormuz, a maritime choke-point for a substantial part of seaborne petroleum traffic \parencite{shepard2020maritime}. Escalations along this corridor act as exogenous supply-side shocks with immediate consequences for certain specific financial instruments \parencite{bouoiyour2019categories}. We refer to the temporal indices of change points in such escalations as \emph{war-markers}; Figures~\ref{fig:crudebrent} and \ref{fig:natgas} plot Brent crude oil futures and Henry Hub natural gas futures, respectively, with the war-markers shown as vertical dashed lines. The measurable influence of these war-markers on diverse instruments is assessed through two asymmetric metrics: \emph{event sensitivity} (ES), the proportion of geopolitical milestones that trigger a structural break, and \emph{event attribution} (EA), the proportion of detected breaks directly attributable to the conflict. Breaks are obtained from a change point detection ensemble built from the top-performing algorithms of Section~\ref{sec:sims} - DCI, CP3O, binary segmentation, MOSUM, and KCPA (all at $\alpha=0.1$) - which records at each time stamp an agreement score for the presence of a change point; high-agreement locations are declared detected breaks. We run the procedure over a diverse sample of 48 financial time series\footnote{The financial market data used in this study are publicly
available from Yahoo Finance (\url{https://finance.yahoo.com/})
and were retrieved programmatically using the
\texttt{tidyquant::tq\_get()} function in R.}: oil and gas futures, producer equities, sector and leveraged ETFs and ETNs, carbon and clean-energy vehicles, and a control block of short-term treasury bills and floating-rate fixed-income instruments, included to check that the resulting exposure structure is coherent rather than arbitrary.

\begin{figure}
    \centering
    \begin{subfigure}{0.325\linewidth}
        \centering
        \includegraphics[width=\linewidth]{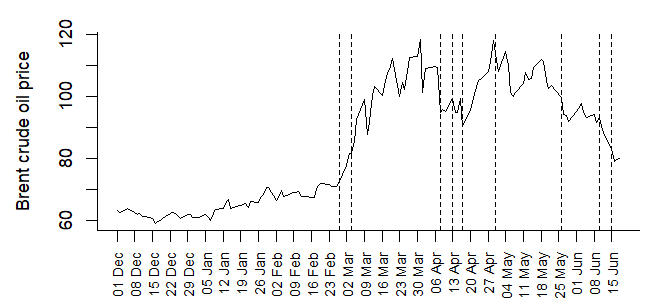}
        \caption{Brent crude oil}
        \label{fig:crudebrent}
    \end{subfigure}
    \begin{subfigure}{0.325\linewidth}
        \centering
        \includegraphics[width=\linewidth]{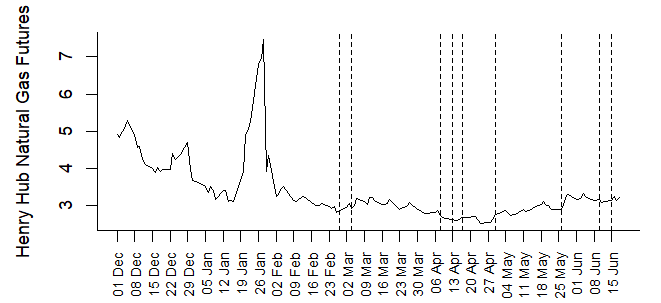}
        \caption{Natural Gas Futures}
        \label{fig:natgas}
    \end{subfigure}
    \begin{subfigure}{0.325\linewidth}
        \centering
        \includegraphics[width=\linewidth]{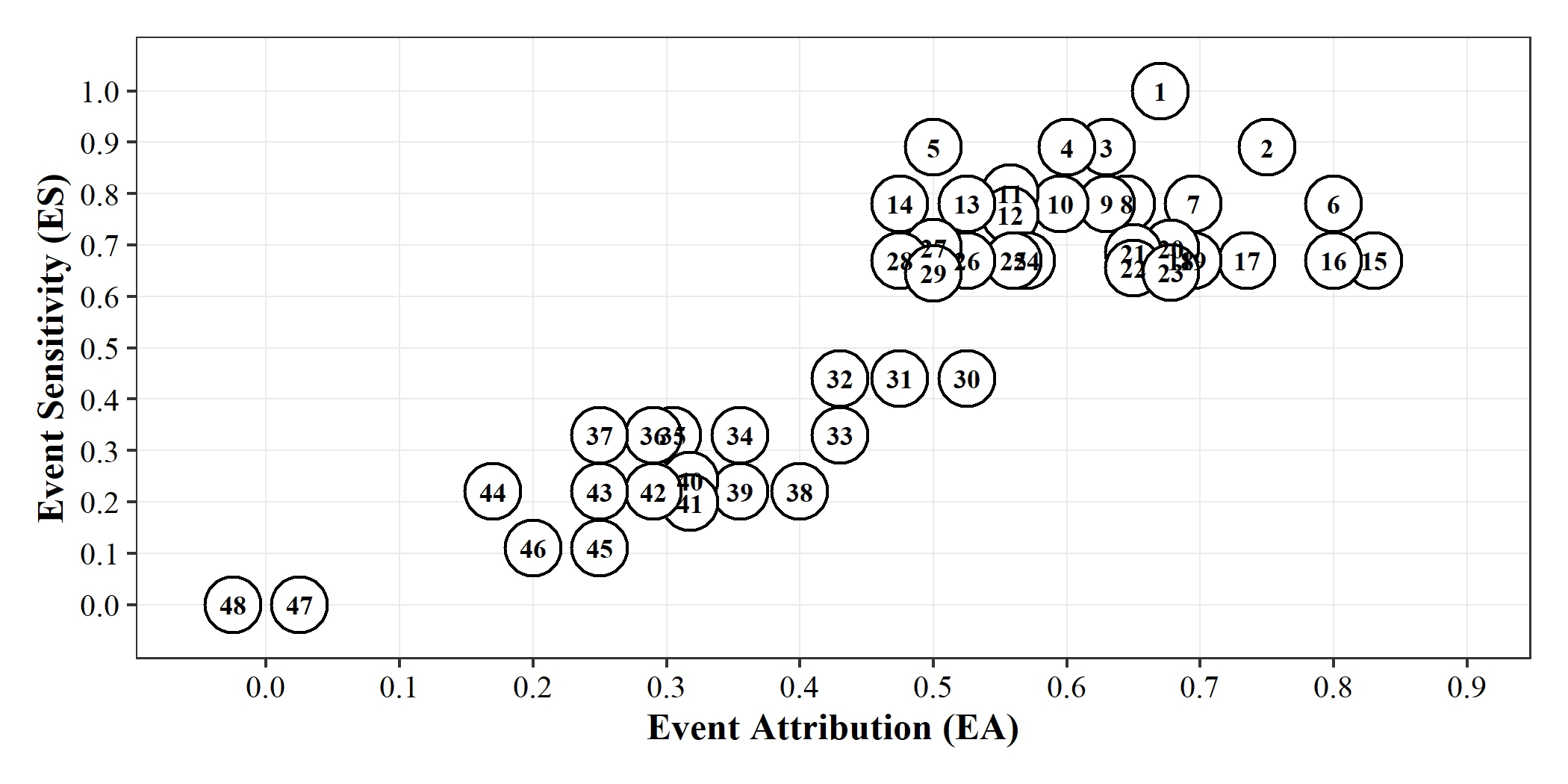}
        \caption{ES-EA clustering}
        \label{fig:oiscatter}
    \end{subfigure}
    \label{fig:impass}
\end{figure}

\begin{table}
\centering
\medskip
\footnotesize
\setlength{\tabcolsep}{4pt}
\begin{tabular}{rlrr@{\hskip 1.5em}rlrr}
\toprule
\textbf{\#} & \textbf{Instrument} & \textbf{ES} & \textbf{EA} &
\textbf{\#} & \textbf{Instrument} & \textbf{ES} & \textbf{EA} \\
\midrule
1  & Fidelity MSCI Energy Index ETF      & 1.00 & 0.67 & 25 & VanEck Oil Refiners ETF              & 0.67 & 0.56 \\
2  & Direxion Daily O\&G E\&P Bull 2X    & 0.89 & 0.75 & 26 & GS Access Treasury 0-1 Yr ETF       & 0.67 & 0.50 \\
3  & Vanguard Energy ETF                 & 0.89 & 0.63 & 27 & Invesco Var.\ Rate Inv.\ Grade ETF   & 0.67 & 0.50 \\
4  & Invesco S\&P 500 EW Energy ETF      & 0.89 & 0.60 & 28 & VanEck Oil Services ETF              & 0.67 & 0.50 \\
5  & Tourmaline Oil                      & 0.89 & 0.50 & 29 & KraneShares Global Carbon ETF        & 0.67 & 0.50 \\
6  & Brent Crude Oil Futures             & 0.78 & 0.80 & 30 & iShares U.S. O\&G E\&P ETF           & 0.44 & 0.50 \\
7  & Exxon Mobil Corporation             & 0.78 & 0.67 & 31 & RBOB Gasoline Futures                & 0.44 & 0.50 \\
8  & Dutch TTF Natural Gas Futures       & 0.78 & 0.67 & 32 & SPDR Bloomberg 1-3 M T-Bill ETF     & 0.44 & 0.43 \\
9  & SPDR S\&P O\&G E\&P ETF             & 0.78 & 0.63 & 33 & Invesco Dyn.\ O\&G Services ETF      & 0.33 & 0.43 \\
10 & Energy Select Sector SPDR Fund      & 0.78 & 0.57 & 34 & US 12 Month Natural Gas Fund         & 0.33 & 0.33 \\
11 & Invesco Dyn.\ Energy E\&P ETF       & 0.78 & 0.57 & 35 & Cenovus Energy                       & 0.33 & 0.33 \\
12 & iShares U.S. Energy ETF             & 0.78 & 0.57 & 36 & Invesco Solar ETF                    & 0.33 & 0.29 \\
13 & Dubai Crude Oil Bear ETN            & 0.78 & 0.50 & 37 & United States Natural Gas Fund       & 0.33 & 0.25 \\
14 & Canadian Natural Resources          & 0.78 & 0.50 & 38 & US Treasury 3 Month Bill ETF         & 0.22 & 0.40 \\
15 & ProShares Ultra Bloomberg Crude     & 0.67 & 0.83 & 39 & Utilities Select Sector SPDR Fund    & 0.22 & 0.33 \\
16 & United States 12 Month Oil Fund     & 0.67 & 0.80 & 40 & American Electric Power Company      & 0.22 & 0.33 \\
17 & United States Oil Fund ETF          & 0.67 & 0.71 & 41 & iShares Treasury FRN ETF             & 0.22 & 0.33 \\
18 & Chord Energy                        & 0.67 & 0.71 & 42 & SPDR O\&G Equip.\ \& Services ETF    & 0.22 & 0.29 \\
19 & Invesco DB Oil Fund                 & 0.67 & 0.67 & 43 & California Carbon Allow.\ Futures    & 0.22 & 0.25 \\
20 & ProShares K-1 Free Crude Oil ETF    & 0.67 & 0.67 & 44 & NextEra Energy, Inc.                 & 0.22 & 0.17 \\
21 & Light Sweet Crude Oil Futures       & 0.67 & 0.67 & 45 & WisdomTree Float.\ Rate Treasury     & 0.11 & 0.25 \\
22 & Heating Oil Futures                 & 0.67 & 0.67 & 46 & ProShares Ultra Bloomberg Nat.\ Gas  & 0.11 & 0.20 \\
23 & E-mini Crude Oil Futures            & 0.67 & 0.67 & 47 & Henry Hub Natural Gas Futures        & 0.00 & 0.00 \\
24 & Dubai Crude Oil Bull ETN            & 0.67 & 0.57 & 48 & E-mini Natural Gas Futures           & 0.00 & 0.00 \\
\bottomrule
\end{tabular}
\caption{ES and EA across the 48 instruments}
\label{tab:eio_metrics}
\end{table}

Table~\ref{tab:eio_metrics} orders the sample by ES, and the ordering tracks logistical and informational proximity to the waterway rather than nominal asset class. At the top sit diversified energy-sector funds, leveraged and broad exploration-and-production vehicles, integrated majors, large seaborne-priced producers, and the futures priced off waterborne barrels; these combine high sensitivity with significant attribution and behave as aggregators of a common macro signal. Equipment and services funds fall lower with weaker attribution, consistent with revenue streams insulated from spot disruption by long-duration contracts. Natural gas splits geographically rather than by instrument type: the European futures contract sits in the leading group, whereas every pipeline-bound North American contract and fund - spot, twelve-month, and leveraged alike - collapses to the floor. Policy-priced and regulated exposures, namely carbon allowances, solar, and regulated utilities, remain weakly affected, as one would expect of instruments priced against legislative rather than maritime timelines. The control group (treasury bills) cluster at the lower end, so the procedure does not indicate exposure where none should exist. Figure~\ref{fig:oiscatter} shows the sample in the ES–EA space, demonstrating the utility of these measures for subsequent supervised and unsupervised learning tasks, such as classification and clustering.
\section{Conclusion}
\label{sec:conc}
In this article, we develop a new test statistic, $I_1$, for detecting a level change in time series, establishing its asymptotic properties and limiting null distribution under i.i.d. Gaussian and weakly dependent settings, and proving consistency under a mean-shift alternative. Building on $I_1$, we propose DCI, a two-phase algorithm for the closely-spaced multiple change points regime that reconciles a global scan with local tests through an adaptive-width sliding window. Simulations show that $I_1$ and DCI remain Pareto-optimal over the TPR-FPR-MAE (Section ~\ref{ssec:pfmi1}) and TPR-FPT-MAE-USCM (Section ~\ref{ssec:pfmdci}) criteria spaces, respectively, with pronounced advantages when breaks occur near sample extremities or in rapid succession. An application to the US-Iran 2026 conflict shows that the financial instruments cluster into a coherent pattern of exposure across energy markets, and that the ES and EA features devised in Section~\ref{sec:app} constitute an information-rich, reduced-dimensional set of variables for time series analysis. Future research will proceed in two directions. First, we will extend $I_1$ to detect changes in trend and in variance, yielding $I_2$ and $I_3$. Second, we will pursue time series tokenization: the embedding of a series of $T$ observations into low-dimensional ES-EA pairs computed against global disruption event sequences, such as wars or Artificial Intelligence model launches. In view of the growing frequency of such disruptions, this tokenization offers a compact and portable representation for assessing their impact on commodities, technology sector indices, or indeed any real-valued temporal sequence of interest.

\printbibliography

@article{truong2020selective,
  title={Selective review of offline change point detection methods},
  author={Truong, Charles and Oudre, Laurent and Vayatis, Nicolas},
  journal={Signal Processing},
  volume={167},
  pages={107299},
  year={2020},
  publisher={Elsevier}
}

@article{aminikhanghahi2017survey,
  title={A survey of methods for time series change point detection},
  author={Aminikhanghahi, Samaneh and Cook, Diane J},
  journal={Knowledge and information systems},
  volume={51},
  number={2},
  pages={339--367},
  year={2017},
  publisher={Springer}
}

@article{page1954continuous,
  title={Continuous inspection schemes},
  author={Page, Ewan S},
  journal={Biometrika},
  volume={41},
  number={1/2},
  pages={100--115},
  year={1954},
  publisher={JSTOR}
}

@article{fearnhead2019detecting,
  title={Detecting changes in slope with an $\boldsymbol{L}_0$ penalty},
  author={Fearnhead, Paul and Maidstone, Robert and Letchford, Adam},
  journal={Journal of Computational and Graphical Statistics},
  volume={28},
  number={2},
  pages={265--275},
  year={2019},
  publisher={Taylor \& Francis}
}

@article{pettitt1979non,
  title={A non-parametric approach to the change-point problem},
  author={Pettitt, Anthony N},
  journal={Journal of the Royal Statistical Society: Series C (Applied Statistics)},
  volume={28},
  number={2},
  pages={126--135},
  year={1979},
  publisher={Wiley Online Library}
}

@article{chen2015graph,
  title={Graph-based change-point detection},
  author={Chen, Hao and Zhang, Nancy},
  journal={The Annals of Statistics},
  volume={43},
  number={1},
  pages={139--176},
  year={2015},
  publisher={Institute of Mathematical Statistics}
}

@article{zou2014nonparametric,
  title={Nonparametric maximum likelihood approach to multiple change-point problems},
  author={Zou, Changliang and Yin, Guosheng and Feng, Long and Wang, Zhaojun},
  journal={The Annals of Statistics},
  volume={42},
  number={3},
  pages={970--1002},
  year={2014},
  publisher={Institute of Mathematical Statistics}
}

@article{levy2009detection,
  title={Detection and localization of change-points in high-dimensional network traffic data},
  author={L{\'e}vy-Leduc, C{\'e}line and Roueff, Fran{\c{c}}ois},
  journal={The Annals of Applied Statistics},
  pages={637--662},
  year={2009},
  publisher={JSTOR}
}

@article{arlot2019kernel,
  title={A kernel multiple change-point algorithm via model selection},
  author={Arlot, Sylvain and Celisse, Alain and Harchaoui, Zaid},
  journal={Journal of machine learning research},
  volume={20},
  number={162},
  pages={1--56},
  year={2019}
}

@article{olshen2004circular,
  title={Circular binary segmentation for the analysis of array-based DNA copy number data},
  author={Olshen, Adam B and Venkatraman, E Seshan and Lucito, Robert and Wigler, Michael},
  journal={Biostatistics},
  volume={5},
  number={4},
  pages={557--572},
  year={2004},
  publisher={Oxford University Press}
}

@article{fryzlewicz2014wild,
  title={Wild binary segmentation for multiple change-point detection},
  author={Fryzlewicz, Piotr},
  journal={The Annals of Statistics},
  volume={42},
  number={6},
  pages={2243--2281},
  year={2014},
  publisher={Institute of Mathematical Statistics}
}

@article{cho2015multiple,
  title={Multiple-change-point detection for high dimensional time series via sparsified binary segmentation},
  author={Cho, Haeran and Fryzlewicz, Piotr},
  journal={Journal of the Royal Statistical Society: Series B (Statistical Methodology)},
  volume={77},
  number={2},
  pages={475--507},
  year={2015},
  publisher={Wiley Online Library}
}

@article{kovacs2023seeded,
  title={Seeded binary segmentation: a general methodology for fast and optimal changepoint detection},
  author={Kov{\'a}cs, Solt and B{\"u}hlmann, Peter and Li, Housen and Munk, Axel},
  journal={Biometrika},
  volume={110},
  number={1},
  pages={249--256},
  year={2023},
  publisher={Oxford University Press}
}

@article{hinkley1970inference,
  title={Inference about the change-point in a sequence of random variables},
  author={Hinkley, David V},
  journal={Biometrika},
  volume={57},
  number={1},
  pages={1--17},
  year={1970},
  publisher={Oxford University Press}
}

@article{inclan1994use,
  title={Use of cumulative sums of squares for retrospective detection of changes of variance},
  author={Inclan, Carla and Tiao, George C},
  journal={Journal of the American Statistical Association},
  volume={89},
  number={427},
  pages={913--923},
  year={1994},
  publisher={Taylor \& Francis}
}

@article{killick2014changepoint,
  title={changepoint: An R package for changepoint analysis},
  author={Killick, Rebecca and Eckley, Idris},
  journal={Journal of statistical software},
  volume={58},
  number={3},
  pages={1--19},
  year={2014}
}

@article{killick2012optimal,
  title={Optimal detection of changepoints with a linear computational cost},
  author={Killick, Rebecca and Fearnhead, Paul and Eckley, Idris A},
  journal={Journal of the American Statistical Association},
  volume={107},
  number={500},
  pages={1590--1598},
  year={2012},
  publisher={Taylor \& Francis}
}

@article{james2015ecp,
  title={ecp: An {R} package for nonparametric multiple change point analysis of multivariate data},
  author={James, Nicholas A and Matteson, David S},
  journal={Journal of Statistical Software},
  volume={62},
  pages={1--25},
  year={2015}
}

@article{baranowski2019narrowest,
  title={Narrowest-over-threshold detection of multiple change points and change-point-like features},
  author={Baranowski, Rafal and Chen, Yining and Fryzlewicz, Piotr},
  journal={Journal of the Royal Statistical Society Series B: Statistical Methodology},
  volume={81},
  number={3},
  pages={649--672},
  year={2019},
  publisher={Oxford University Press}
}

@article{maeng2024detecting,
  title={Detecting linear trend changes in data sequences},
  author={Maeng, Hyeyoung and Fryzlewicz, Piotr},
  journal={Statistical Papers},
  volume={65},
  number={3},
  pages={1645--1675},
  year={2024},
  publisher={Springer}
}

@inproceedings{zhang2017pruning,
  title={Pruning and nonparametric multiple change point detection},
  author={Zhang, Wenyu and James, Nicholas A and Matteson, David S},
  booktitle={2017 IEEE international conference on data mining workshops (ICDMW)},
  pages={288--295},
  year={2017},
  organization={IEEE}
}

@article{scott1974cluster,
  title={A cluster analysis method for grouping means in the analysis of variance},
  author={Scott, Andrew Jhon and Knott, M},
  journal={Biometrics},
  pages={507--512},
  year={1974},
  publisher={JSTOR}
}

@article{srivastava1986likelihood,
  title={Likelihood ratio tests for a change in the multivariate normal mean},
  author={Srivastava, MS and Worsley, Keith J},
  journal={Journal of the American Statistical Association},
  volume={81},
  number={393},
  pages={199--204},
  year={1986},
  publisher={Taylor \& Francis}
}

@article{fearnhead2019changepoint,
  title={Changepoint detection in the presence of outliers},
  author={Fearnhead, Paul and Rigaill, Guillem},
  journal={Journal of the American Statistical Association},
  volume={114},
  number={525},
  pages={169--183},
  year={2019},
  publisher={Taylor \& Francis}
}

@article{auger1989algorithms,
  title={Algorithms for the optimal identification of segment neighborhoods},
  author={Auger, Ivan E and Lawrence, Charles E},
  journal={Bulletin of mathematical biology},
  volume={51},
  number={1},
  pages={39--54},
  year={1989},
  publisher={Elsevier}
}

@article{matteson2014nonparametric,
  title={A nonparametric approach for multiple change point analysis of multivariate data},
  author={Matteson, David S and James, Nicholas A},
  journal={Journal of the American Statistical Association},
  volume={109},
  number={505},
  pages={334--345},
  year={2014},
  publisher={Taylor \& Francis}
}

@misc{wang2015bayesian,
  title={Bayesian change point analysis of linear models on graphs},
  author={Wang, Xiaofei and Emerson, John W},
  note={\textit{arXiv:1509.00817}},
  year={2015}
}

@misc{miller2018cpat,
  title={{CPAT}: {C}hange point analysis tests},
  author={Miller, Curtis},
  note={\textit{https://CRAN.R-project.org/package=CPAT}},
  year={2018}
}

@article{pohlert2020trend,
  title={Trend: non-parametric trend tests and change-point detection},
  author={Pohlert, Thorsten},
  journal={R package version},
  volume={1},
  number={4},
  year={2020}
}

@misc{wang2018package,
  title={Package ‘bcp’},
  author={Wang, Xiaofei and Erdman, Chandra and Emerson, John W and Wang, Maintainer Xiaofei and DNAcopy, Suggests},
  year={2018},
  publisher={Version}
}

@article{beaulieu2018distinguishing,
  title={Distinguishing trends and shifts from memory in climate data},
  author={Beaulieu, Claudie and Killick, Rebecca},
  journal={Journal of Climate},
  volume={31},
  number={23},
  pages={9519--9543},
  year={2018},
  publisher={American Meteorological Society}
}

@article{baranowski2019not,
  title={not: Narrowest-over-threshold change-point detection},
  author={Baranowski, R and Chen, Y and Fryzlewicz, P},
  journal={R package version},
  volume={1},
  year={2019}
}

@misc{killick2021envcpt,
  title={EnvCpt: Detection of structural changes in climate and environment time series. {R} package version 1.1.3},
  author={Killick, R and Beaulieu, C and Taylor, S and Hullait, H},
  year={2021}
}

@article{waychal2023adaptive,
  title={An adaptive multi-objective optimal forecast combination and its application for predicting intermittent demand},
  author={Waychal, Nachiketas and Laha, Arnab Kumar and Sinha, Ankur},
  journal={Journal of the Operational Research Society},
  volume={75(9)},
  pages={1--13},
  year={2023},
  publisher={Taylor \& Francis}
}

@article{londschien2021change,
  title={Change-point detection for graphical models in the presence of missing values},
  author={Londschien, Malte and Kov{\'a}cs, Solt and B{\"u}hlmann, Peter},
  journal={Journal of Computational and Graphical Statistics},
  volume={30},
  number={3},
  pages={768--779},
  year={2021},
  publisher={Taylor \& Francis}
}

@article{drikvandi2025distribution,
  title={A distribution-free method for change point detection in non-sparse high dimensional data},
  author={Drikvandi, Reza and Modarres, Reza},
  journal={Journal of Computational and Graphical Statistics},
  volume={34},
  number={1},
  pages={290--305},
  year={2025},
  publisher={Taylor \& Francis}
}

@article{berman1964,
  author  = {Berman, Simeon M.},
  title   = {Limit theorems for the maximum term in stationary sequences},
  journal = {Annals of Mathematical Statistics},
  volume  = {35},
  number  = {2},
  pages   = {502--516},
  year    = {1964},
  doi     = {10.1214/aoms/1177703551}
}

@article{bradley2005,
  author  = {Bradley, Richard C.},
  title   = {Basic properties of strong mixing conditions. {A} survey and some open questions},
  journal = {Probability Surveys},
  volume  = {2},
  pages   = {107--144},
  year    = {2005},
  doi     = {10.1214/154957805100000104}
}

@article{davydov1968,
  author  = {Davydov, Yu. A.},
  title   = {Convergence of distributions generated by stationary stochastic processes},
  journal = {Theory of Probability and Its Applications},
  volume  = {13},
  number  = {4},
  pages   = {691--696},
  year    = {1968},
  doi     = {10.1137/1113086}
}

@article{eichinger2018,
  author  = {Eichinger, Birte and Kirch, Claudia},
  title   = {A {MOSUM} procedure for the estimation of multiple random change points},
  journal = {Bernoulli},
  volume  = {24},
  number  = {1},
  pages   = {526--564},
  year    = {2018},
  doi     = {10.3150/16-BEJ887}
}

@article{fisher1928,
  author  = {Fisher, Ronald A. and Tippett, Leonard H. C.},
  title   = {Limiting forms of the frequency distribution of the largest or smallest member of a sample},
  journal = {Mathematical Proceedings of the Cambridge Philosophical Society},
  volume  = {24},
  number  = {2},
  pages   = {180--190},
  year    = {1928},
  doi     = {10.1017/S0305004100015681}
}

@article{gnedenko1943,
  author  = {Gnedenko, Boris},
  title   = {Sur la distribution limite du terme maximum d'une s{\'e}rie al{\'e}atoire},
  journal = {Annals of Mathematics},
  volume  = {44},
  number  = {3},
  pages   = {423--453},
  year    = {1943},
  doi     = {10.2307/1968974}
}

@book{gumbel1958,
  author    = {Gumbel, Emil J.},
  title     = {Statistics of Extremes},
  publisher = {Columbia University Press},
  address   = {New York},
  year      = {1958}
}

@article{leadbetter1974,
  author  = {Leadbetter, M. Ross},
  title   = {On extreme values in stationary sequences},
  journal = {Zeitschrift f{\"u}r Wahrscheinlichkeitstheorie und Verwandte Gebiete},
  volume  = {28},
  pages   = {289--303},
  year    = {1974},
  doi     = {10.1007/BF00532947}
}

@article{leadbetter1983extremal,
  author  = {Leadbetter, M. Ross},
  title   = {Extremes and local dependence in stationary sequences},
  journal = {Zeitschrift f{\"u}r Wahrscheinlichkeitstheorie und Verwandte Gebiete},
  volume  = {65},
  pages   = {291--306},
  year    = {1983},
  doi     = {10.1007/BF00532484}
}

@book{llr1983,
  author    = {Leadbetter, M. Ross and Lindgren, Georg and Rootz{\'e}n, Holger},
  title     = {Extremes and Related Properties of Random Sequences and Processes},
  series    = {Springer Series in Statistics},
  publisher = {Springer},
  address   = {New York},
  year      = {1983},
  doi       = {10.1007/978-1-4612-5449-2}
}

@article{rosenblatt1956,
  author  = {Rosenblatt, Murray},
  title   = {A central limit theorem and a strong mixing condition},
  journal = {Proceedings of the National Academy of Sciences of the USA},
  volume  = {42},
  number  = {1},
  pages   = {43--47},
  year    = {1956},
  doi     = {10.1073/pnas.42.1.43}
}

@article{watson1954,
  author  = {Watson, Geoffrey S.},
  title   = {Extreme values in samples from {$m$}-dependent stationary stochastic processes},
  journal = {Annals of Mathematical Statistics},
  volume  = {25},
  number  = {4},
  pages   = {798--800},
  year    = {1954}
}

@misc{van2020evaluation,
  title={An evaluation of change point detection algorithms},
  author={Van den Burg, Gerrit JJ and Williams, Christopher KI},
  note={\textit{arXiv:2003.06222}},
  year={2020}
}

@article{fremdt2015page,
  title={Page's sequential procedure for change-point detection in time series regression},
  author={Fremdt, Stefan},
  journal={Statistics},
  volume={49},
  number={1},
  pages={128--155},
  year={2015},
  publisher={Taylor \& Francis}
}

@article{qin2017strong,
  title={A strong convergence rate of estimator of variance change in linear processes and its applications},
  author={Qin, Ruibing and Liu, Weiqi and Tian, Zheng},
  journal={Statistics},
  volume={51},
  number={2},
  pages={314--330},
  year={2017},
  publisher={Taylor \& Francis}
}

@article{han2023optimization,
  title={An optimization method for change-point monitoring in finite samples sequence},
  author={Han, Dong and Tsung, Fugee and Xian, Jinguo},
  journal={Statistics},
  volume={57},
  number={5},
  pages={1251--1266},
  year={2023},
  publisher={Taylor \& Francis}
}

@article{shepard2020maritime,
  title={Maritime piracy in the Strait of Hormuz and implications of energy export security},
  author={Shepard, Jun U and Pratson, Lincoln F},
  journal={Energy Policy},
  volume={140},
  pages={111379},
  year={2020},
  publisher={Elsevier}
}

@article{bouoiyour2019categories,
  title={What are the categories of geopolitical risks that could drive oil prices higher? {A}cts or threats?},
  author={Bouoiyour, Jamal and Selmi, Refk and Hammoudeh, Shawkat and Wohar, Mark E},
  journal={Energy Economics},
  volume={84},
  pages={104523},
  year={2019},
  publisher={Elsevier}
}

\appendix
\section{Proofs of theorems}
\label{apsec:proofs}
\subsection{Proof of Proposition~\ref{prop:stat}}
\label{prfp1}
\begin{proof}
Since $y_t\sim N(0,1)$, $E[z_t]=\sum_{r\in\mathcal{C}_w}c_{r,w}E[y_{t+r}]=0$. Because the coefficients $c_{r,w}$ are independent of the time index $t$, shifting the index from $t$ to $t+k$ translates the underlying coefficients over identical relative distances. Thus, the joint distribution of $(z_{t_1},\dots,z_{t_m})$ is identical to that of $(z_{t_1+k},\dots,z_{t_m+k})$, establishing strict stationarity. Further, $\mathrm{Var}(z_t)=\mathrm{Var}\big(\sum_{r\in\mathcal{C}_w}c_{r,w}y_{t+r}\big)=\sum_{r=-2w+2}^{2w-1}c_{r,w}^2$. Substituting $c_{r,w}$ from Equation~\eqref{eq:crw}, $\mathrm{Var}(z_t)=\frac{4w}{3}+\frac{2}{3w}$.

For $k\in\mathbb{Z}$, the autocovariance function is $\gamma_z(k)=\mathrm{Cov}(z_t,z_{t+k})=\sum_{r\in\mathcal{C}_w}c_{r,w}c_{r-k,w}$ where $c_{r,w}=0$ outside $\mathcal{C}_w$. Therefore, $\rho_z(k)=\gamma_z(k)/\sigma_z^2$ is independent of $t$. Moreover, $\gamma_z(k)=0$ for all $|k|>4w-3$. Thus, $z_t$ is $m$-dependent where $m=4w-3$. Since $|z_t|=g(z_t)$ for the fixed measurable function $g(x)=|x|$, and both strict stationarity and $m$-dependence are preserved under any measurable pointwise (instantaneous) transformation of the sequence, $\{|z_t|\}$ is strictly stationary and $m$-dependent.
\end{proof}

\subsection{Proof of Lemma~\ref{lem:berman}}
\label{prflem1}
\begin{proof}
Under $H_0$ and the conditions of Proposition~\ref{prop:stat}, let the sequence of normalizing constants be defined as \parencite[Theorem~1.5.3]{llr1983}:
\begin{equation}
a_N=\frac{\sigma_z}{\sqrt{2\log N}}
\quad\text{and}\quad
b_N=\sigma_z\left(\sqrt{2\log N}-\frac{\log(\log N)+\log(4\pi)}{2\sqrt{2\log N}}\right). \label{eq:norm}
\end{equation}
Berman's condition requires that $\rho_z(k)\log k\to 0$ as $k\to\infty$. For any fixed lag $k$, Proposition~\ref{prop:stat} ensures that $\gamma_z(k)$ drops to zero for a finite $k$, satisfying Berman's condition. Leadbetter's $D(u_N)$ long-range mixing condition requires that the mixing coefficient $\alpha_z(k)$ vanishes asymptotically for distant sets. Because the process $z_t$ is $m$-dependent, $\alpha_z(k)=0$ for all lags $k>m=4w-3$, so $D(u_N)$ holds automatically \parencite{watson1954,leadbetter1974}.

To verify the local anti-clustering condition $D'(u_N)$, we must show that the probability of multiple high-level exceedances occurring within a close neighborhood vanishes: $\lim_{N\to\infty}N\sum_{j=2}^{p_N}P(z_1>u_N,\,z_j>u_N)=0$, where $p_N=o(N)$. For any $j\neq 1$ within the $m$-dependence window, the joint distribution of $(z_1,z_j)$ is bivariate normal with correlation coefficient $\rho_z(j-1)<1$. Standard bounds on bivariate normal tail probabilities \parencite{berman1964,llr1983} state that $P(z_1>u_N,\,z_j>u_N)\le C\exp\!\big(-\frac{u_N^2}{1+\rho_z(j-1)}\big)$. Since $u_N^2\sim 2\log N$, this probability scales as $O(N^{-2/(1+\rho)})$. Because $\rho_z(j-1)<1$, the exponent $-2/(1+\rho)$ is strictly less than $-1$. Summing over the $m$-dependent steps (where $m=O(\log N)$) and multiplying by $N$ yields a negligible term that vanishes as $N\to\infty$. This absence of local cascading implies that extreme values appear in isolation, fixing the extremal index of $z_t$ to $\theta=1$ \parencite{leadbetter1983extremal}. Finally, since $|z_t|$ inherits strict stationarity and $m$-dependence, $\alpha_{|z|}(k)=0$ for $k>m$ and $D(u_N)$ holds automatically; moreover, as $\{|z_1|>u_N,\,|z_j|>u_N\}$ is the union of the four joint tail events $\{\pm z_1>u_N,\,\pm z_j>u_N\}$, each bounded by the same bivariate normal tail estimate above (with correlation $\pm\rho_z(j-1)$, and $|\rho_z(j-1)|<1$), the $D'(u_N)$ sum for $\{|z_t|\}$ is at most four times that for $\{z_t\}$ and likewise vanishes, whence $\{|z_t|\}$ also has extremal index $\theta=1$.
\end{proof}

\subsection{Proof of Theorem~\ref{thm:gumbel}}
\label{prfthm1}
\begin{proof}
For a stationary Gaussian sequence satisfying Berman's condition with variance $\sigma_z^2$ and an extremal index $\theta=1$, classical extreme value theory \parencite{berman1964,llr1983} states that the unconstrained maximum $M_N=\max_{t\in\mathcal{I}_T}z_t$ satisfies $\lim_{N\to\infty}P\big((M_N-b_N)/a_N\le x\big)=\exp(-e^{-x})$. Our test statistic, however, is defined by the absolute maximum $I_1=\max_{t\in\mathcal{I}_T}|z_t|=\max(\max z_t,\,\max(-z_t))$. This requires looking at both the upper and lower tails of the process.

Let $u_N=a_Nx+b_N$ represent the diverging threshold level. The probability that the absolute maximum remains below this threshold can be written as $P_{H_0}(I_1\le u_N)=P_{H_0}(\max_{t\in\mathcal{I}_T}z_t\le u_N,\,\min_{t\in\mathcal{I}_T}z_t\ge -u_N)$. Because $u_N\to\infty$ as $N\to\infty$, high-level exceedances in the positive and negative directions become asymptotically independent events \parencite{llr1983}. Due to the symmetry of the Gaussian distribution, the marginal probabilities of the two tails are identical. Applying the product rule for independent extreme events yields $\lim_{N\to\infty}P_{H_0}(I_1\le u_N)=\lim_{N\to\infty}P_{H_0}(\max_{t\in\mathcal{I}_T}z_t\le u_N)\times P_{H_0}(\max_{t\in\mathcal{I}_T}(-z_t)\le u_N)=\exp(-e^{-x})\cdot\exp(-e^{-x})=\exp(-2e^{-x})$, establishing Equation~\eqref{eq:limit}.

To derive the moments, let $Y_N=(I_1-b_N)/a_N$. The limiting cumulative distribution function (CDF) is $F_Y(y)=\exp(-2e^{-y})=\exp(-e^{-(y-\log 2)})$. This corresponds to a standard Gumbel distribution shifted by a location parameter of $\log 2$. The mean of a standard Gumbel random variable is equal to $\gamma$ \parencite{gumbel1958}. By the linearity of expectations, the mean of the shifted variable is $E[Y_N]\to\gamma+\log 2$. Transforming this back to the original scale yields Equation~\eqref{eq:mean}. Because a location shift does not affect the dispersion, the variance of $Y_N$ converges directly to the standard Gumbel variance, $\pi^2/6$. Rescaling by the normalization factor $a_N^2$ yields Equation~\eqref{eq:var}.
\end{proof}

\subsection{Proof of Corollary~\ref{cor:crit}}
\label{prfcrl1}
\begin{proof}
Let $\lambda_\alpha$ denote the asymptotic level-$\alpha$ critical value for $I_1$ such that $\lim_{N\to\infty}P_{H_0}(I_1>\lambda_\alpha)=\alpha$, which is equivalent to the threshold condition $\lim_{N\to\infty}P_{H_0}(I_1\le\lambda_\alpha)=1-\alpha$. We express the critical value on the standardized scale by setting $\lambda_\alpha=a_Nx_\alpha+b_N$, where $x_\alpha$ is the $1-\alpha$ quantile of the limiting law. Substituting this into the Gumbel limit from Equation~\eqref{eq:limit} yields $\exp(-2e^{-x_\alpha})=1-\alpha$. Therefore, $x_\alpha=-\log\left(-\frac{1}{2}\log(1-\alpha)\right)$. Substituting this expression back into the threshold scaling equation $\lambda_\alpha=b_N+a_Nx_\alpha$ yields the operational decision threshold given in Equation~\eqref{eq:crit}.
\end{proof}

\subsection{Proof of Lemma~\ref{lem:mixing}}
\label{prflem2}
\begin{proof}
Because $z_t$ is constructed as a time-invariant measurable mapping of a finite block of variables $(y_{t-2w+2},\dots,y_{t+2w-1})$, any temporal translation shifts the entire block uniformly. Strict stationarity of $z_t$ follows directly from the strict stationarity of $y_t$. To evaluate the mixing coefficients, let $\mathcal{M}_a^b(y)$ and $\mathcal{M}_a^b(z)$ denote the $\sigma$-algebras generated by the respective histories \parencite{bradley2005}. The variable $z_t$ is fully determined by the variables up to index $t+2w-1$, meaning $\mathcal{M}_{-\infty}^{t}(z)\subseteq\mathcal{M}_{-\infty}^{t+2w-1}(y)$. Similarly, for a future index $t+k$, the variable $z_{t+k}$ depends on elements starting from $t+k-2w+2$, meaning $\mathcal{M}_{t+k}^{\infty}(z)\subseteq\mathcal{M}_{t+k-2w+2}^{\infty}(y)$. The mixing coefficient measures the maximum dependence between these $\sigma$-algebras:
\begin{align*}
\alpha_z(k)&=\sup|P(A\cap B)-P(A)P(B)|\quad\text{for }A\in\mathcal{M}_{-\infty}^{t}(z),\ B\in\mathcal{M}_{t+k}^{\infty}(z)\\
&\le\sup|P(A'\cap B')-P(A')P(B')|\quad\text{for }A'\in\mathcal{M}_{-\infty}^{t+2w-1}(y),\ B'\in\mathcal{M}_{t+k-2w+2}^{\infty}(y).
\end{align*}
The temporal separation between these two underlying segments of $y$ is $(t+k-2w+2)-(t+2w-1)=k-4w+3$. This establishes the inequality in Equation~\eqref{eq:mix}.

The variance expression in Equation~\eqref{eq:vardep} follows by expanding the variance of a linear combination: $\mathrm{Var}(\sum c_{r,w}y_{t+r})=\sum\sum c_{r,w}c_{s,w}\mathrm{Cov}(y_{t+r},y_{t+s})$. Because $\alpha_y(k)=O(\rho^k)$, mixing inequalities \parencite{davydov1968} imply that the autocovariances decay exponentially ($\gamma_y(k)=O(\rho^k)$). Since the coefficients $c_{r,w}$ are bounded and the double sum is restricted to the compact support of the filter $\mathcal{C}_w$, the summation is stable and grows linearly with the window width $w$. In equations and inequalities:
\[
\sigma_z^2=\sum_{r\in\mathcal{C}_w}\sum_{s\in\mathcal{C}_w}c_{r,w}c_{s,w}\gamma_y(s-r)=\sum_{h\in\mathbb{Z}}\gamma_y(h)\sum_{r\in\mathcal{C}_w}c_{r,w}c_{r+h,w}.
\]
By Cauchy's inequality,
\[
\left|\sum_{r\in\mathcal{C}_w}c_{r,w}c_{r+h,w}\right|\le\sum_{r\in\mathcal{C}_w}c_{r,w}^2.
\]
From Proposition~\ref{prop:stat}, $\sum_{r\in\mathcal{C}_w}c_{r,w}^2=\frac{4w}{3}+\frac{2}{3w}=O(w)$. Further, $\sum_h|\gamma_y(h)|<\infty$ because $\alpha_y(k)=O(\rho^k)$. Therefore,
\[
\sigma_z^2\le\left(\sum_{h\in\mathbb{Z}}|\gamma_y(h)|\right)\left(\sum_{r\in\mathcal{C}_w}c_{r,w}^2\right)=O(w).
\]
The rapid exponential decay of the serial correlations satisfies Berman's condition under dependence \parencite{berman1964} (thus satisfying Leadbetter's conditions as expanded in Lemma~\ref{lem:berman}), preventing extreme values from clustering and ensuring the extremal index of $z_t$ (and consequently $|z_t|$) remains $\theta=1$ \parencite{leadbetter1983extremal}.
\end{proof}

\subsection{Proof of Theorem~\ref{thm:dep}}
\label{prfthm2}
\begin{proof}
Lemma~\ref{lem:mixing} demonstrates that the extremal index of $z_t$ (and consequently $|z_t|$) is preserved at $\theta=1$. By applying extreme value theorems for stationary processes with short-range dependence \parencite{leadbetter1974,llr1983} (Theorem~\ref{thm:gumbel}), the unconstrained maximum of $z_t$ converges asymptotically to a standard Gumbel distribution. Because the threshold $u_N\to\infty$, the positive and negative tail exceedances remain asymptotically independent. Combining the symmetric positive and negative tail exceedance probabilities yields the same two-sided Gumbel distribution $\exp(-2e^{-x})$. This confirms that the extreme value calibration (Theorem~\ref{thm:gumbel}), the normalizing constants (Lemma~\ref{lem:berman}), and the operational threshold (Corollary~\ref{cor:crit}) remain invariant under short-range serial dependence (Assumption~\ref{as:dep}).
\end{proof}

\subsection{Proof of Lemma~\ref{lem:signal}}
\label{prflem3}
\begin{proof}
Using Equation \eqref{eq:zt}, $m_z(t)=E(z_t)=\sum_{r\in\mathcal{C}_w}c_{r,w}\mu_{t+r}$. For the single-shift mean function defined in Assumption~\ref{as:alt}, $\mu_{t+r}=\mu_1\mathbf 1\{t+r<\eta\}+\mu_2\mathbf 1\{t+r\ge\eta\}$. Since the filter has zero total mass, $\sum_{r\in\mathcal{C}_w}c_{r,w}=0$, the summation $\sum_{r\in\mathcal{C}_w}c_{r,w}\mu_{t+r}$ equals $0$ for all $t$ outside the symmetric temporal neighborhood $\mathcal{C}_w$ centered at $t=\eta$. Hence $m_z(t)=\delta\sum_{r\in\mathcal{C}_w}c_{r,w}\mathbf 1\{t+r\ge\eta\}$. Let $u=t-\eta$ and define
\begin{equation}
H_w(u)=\sum_{r\in\mathcal{C}_w}c_{r,w}\mathbf 1\{u+r\ge 0\}. \label{eq:Hw}
\end{equation}
Then $m_z(t)=\delta H_w(t-\eta)$. The response is localized because $c_{r,w}=0$ outside $\mathcal{C}_w$ (Equation~\eqref{eq:crw}).

It remains to compute the peak height. The positive coefficients of the filter occur for $r=1,\dots,2w-1$. Their sum is
\begin{equation}
\sum_{r=1}^{2w-1}c_{r,w}=\frac{1}{w}\left[\sum_{r=1}^{w}r+\sum_{r=w+1}^{2w-1}(2w-r)\right]=\frac{1}{w}\left[\frac{w(w+1)}{2}+\frac{w(w-1)}{2}\right]=w. \label{eq:possum}
\end{equation}
The negative coefficients have total mass $-w$, since $\sum_{r\in\mathcal{C}_w}c_{r,w}=0$. Therefore, no cumulative response can exceed $w$ in absolute value, and this value is attained when the cumulative sum includes exactly the positive part of the filter, i.e., at $t=\eta$ or $u=0$. Hence $\max_t|m_z(t)|=w\delta$.
\end{proof}

\subsection{Proof of Theorem~\ref{thm:power}}
\label{prfthm3}
\begin{proof}
Under $H_1$,
\begin{equation}
z_t=m_z(t)+\xi_t,\qquad \xi_t=z_t-m_z(t). \label{eq:decomp}
\end{equation}
By Lemma~\ref{lem:signal}, $\max_t|m_z(t)|=M_N=w\delta$. Let $t^\star$ be an index at which this deterministic maximum is attained. Then
\[
I_1=\max_{t\in\mathcal{I}_T}|z_t|\ge|z_{t^\star}|=|m_z(t^\star)+\xi_{t^\star}|.
\]
Since $|m_z(t^\star)|=M_N$, $|z_{t^\star}|\ge M_N-|\xi_{t^\star}|$. Therefore,
\begin{align}
P_{H_1}(I_1\le\lambda_\alpha)&\le P\left(M_N-|\xi_{t^\star}|\le\lambda_\alpha\right)\notag\\
&=P\left(|\xi_{t^\star}|\ge M_N-\lambda_\alpha\right). \label{eq:bound}
\end{align}
Under either Assumption~\ref{as:iid} or Assumption~\ref{as:dep}, $\xi_{t^\star}$ has variance of order $\sigma_z^2$, and therefore $\frac{\xi_{t^\star}}{\sigma_z}=O_p(1)$.

Now we verify the order of $M_N-\lambda_\alpha$ under asymptotic conditions $N\to\infty$, $w\to\infty$, and $w/\log N\to K$. Since $\sigma_z^2=\frac{4w}{3}+\frac{2}{3w}$,
\[
\sigma_z\sqrt{2\log N}=\sqrt{\left(\frac{4w}{3}+\frac{2}{3w}\right)2\log N}=\sqrt{\frac{8w\log N}{3}}\,\{1+o(1)\}.
\]
Using Equations \eqref{eq:norm} and \eqref{eq:crit}, $\lambda_\alpha=\sigma_z\sqrt{2\log N}+o(\sqrt{w\log N})$. Therefore,
\begin{equation}
M_N-\lambda_\alpha=w\delta-\sqrt{\frac{8w\log N}{3}}\,\{1+o(1)\}=\log N\left[\delta K-\sqrt{\frac{8K}{3}}+o(1)\right]. \label{eq:gap}
\end{equation}
Thus $\lim_{N\to\infty}M_N-\lambda_\alpha\to\infty$ whenever $K>\frac{8}{3\delta^2}$. Using this condition in Equation \eqref{eq:bound},
\begin{align}
\lim_{N\to\infty}P_{H_1}(I_1\le\lambda_\alpha)&\le P\left(|\xi_{t^\star}|\ge M_N-\lambda_\alpha\right)\to 0 \label{eq:vanish}\\
\therefore\quad \lim_{N\to\infty}P_{H_1}(I_1>\lambda_\alpha)&\to 1 \label{eq:consist}
\end{align}
\end{proof}

\section{Design enhancements in $I_1$ and $DCI$}
\label{apsec:enhance}
\subsection{Variance estimation}
\label{ssec:varest}
The input data series $y_{1_t}$ may contain change point(s), thus using $\sigma_{y_1} = \sqrt{ \frac{1}{T}\sum_{t=1}^{T}(y_{1_t}-\bar{y_1})^2}$ can overestimate the underlying within-segment variance. To prevent this, we invoke a combination of properties - differencing and inter-quartile range - for robust estimation of variance as:
$\sigma_{y_1} = \frac{IQR(\Delta y_{1_t})}{-2\sqrt{2}\Phi^{-1}(0.25)}$ (under i.i.d. Gaussian assumption) where $IQR(\Delta y_{1_t})$ is the inter-quartile range of $\Delta y_{1_t} = y_{1_t} - y_{1_{t-1}}$. Under the assumption of weak dependence, the same is estimated as $\sigma_{y_1} = \frac{IQR(\Delta y_{1_t})}{-2\sqrt{2}(1-\rho)\Phi^{-1}(0.25)}$ where $\rho \in (-1,1)$.
\subsection{Adaptive width sliding window}
\label{ssec:adapwid}
To prevent boundary information loss, $d_t$ is enhanced by adaptively scaling the sliding window width to available data:
\begin{equation}
\label{eqadapwin}
d_t = \frac{1}{w_t^+} \sum_{i=1}^{w_t^+} y_{t+i} - \frac{1}{w_t^-} \sum_{i=1}^{w_t^-} y_{t-i+1}, \quad \text{for } t \in \{1, \dots, T-1\}
\end{equation}
where $w_t^- = \min(w, t)$ and $w_t^+ = \min(w, T-t)$ denote the left and right adaptive window sizes, respectively. This expands the admissible index set $\mathcal{I}_T$ to $\{1,\dots,T-1\}$, ensuring robust performance under limited data availability.

\subsection{$p$-local scaling}
\label{plocal}
Let $\mathbf{H} \in \mathbb{R}^{p \times (T-p)}$ be the Hankel matrix constructed from $\{d_t\}_{t=1}^{T-1}$ with entries $h_{ij} = d_{i+j-1}$. For each row $\mathbf{h}_{i\cdot}$ ($i=1,\dots,p$), let $Q_{1,i}$ and $Q_{3,i}$ denote its $25^{\text{th}}$ and $75^{\text{th}}$ percentiles, and define $\text{IQR}_i = Q_{3,i} - Q_{1,i}$. The IQR-exceedance transformation $h^{\prime}_{ij} = d^{\prime}_{i+j-1|i}$ is defined as:
\begin{equation}
\label{eqplocal1}
d^{\prime}_{i+j-1|i} = \frac{[h_{ij} - Q_{3,i}]_+ - [Q_{1,i} - h_{ij}]_+}{\text{IQR}_i}
\end{equation}
where $[\cdot]_+ = \max(0, \cdot)$. The $p$-locally scaled series $dp_t$ is obtained by averaging across valid trajectory estimates: $dp_t = \frac{1}{|\mathcal{A}_t|} \sum_{i \in \mathcal{A}_t} d^{\prime}_{t|i}$, where the active row index set $\mathcal{A}_t$ for time $t \in \{1, \dots, T-1\}$ is given by:
\begin{equation}
\label{eqplocal2}
\mathcal{A}_t = \{i \in \{1,\dots,p\} \mid 1 \le t - i + 1 \le T-p\} = \{\max(1, t - T + p + 1), \dots, \min(p, t)\}
\end{equation}
\end{document}